\documentclass[sigconf, screen, nonacm, authorversion]{acmart}
\AtBeginDocument{%
  }

\usepackage{amsthm}
\usepackage{amsmath}
\usepackage{mathtools}
\usepackage{algorithm}
\usepackage{algpseudocode}
\usepackage{dblfloatfix}
\usepackage{xspace}
\usepackage{nicefrac}
\usepackage{subcaption}
\usepackage{multirow}

\newtheorem{theorem}{Theorem}[section]
\newtheorem{lemma}[theorem]{Lemma}

\DeclareMathOperator*{\argmin}{arg\,min}
\begin{document}

\title[Signal-Aware Workload Shifting Algorithms with UQ Predictors]{Signal-Aware Workload Shifting Algorithms \\with Uncertainty-Quantified Predictors}

\newcommand{\sref}[2]{\hyperref[#2]{#1 \ref{#2}}}

\renewcommand{\sectionautorefname}{Section}
\renewcommand{\subsectionautorefname}{Section}
\renewcommand{\subsubsectionautorefname}{Section}
\renewcommand{\figureautorefname}{Figure}
\def\algorithmautorefname{Algorithm}
\def\thmautorefname{Theorem}
\def\claimautorefname{Claim}
\def\lemautorefname{Lemma}
\def\appendixautorefname{Appendix}

\newcommand{\OCS}{\texttt{OCS}\xspace}
\newcommand{\OCU}{\texttt{SASP}\xspace}
\newcommand{\ROCU}{\texttt{SASP}\xspace}
\newcommand{\SWAP}{\texttt{SASP}\xspace}
\newcommand{\SASP}{\texttt{SASP}\xspace}
\newcommand{\ROAdv}{\texttt{RO-Advice}\xspace}
\newcommand{\UQAdv}{\texttt{UQ-Advice}\xspace}
\newcommand{\ALG}{\texttt{ALG}\xspace}
\newcommand{\OPT}{\texttt{OPT}\xspace}
\newcommand{\ADV}{\texttt{ADV}\xspace}
\newcommand{\DUS}{\texttt{DUS}\xspace}
\newcommand{\RORO}{\texttt{RORO}\xspace}
\newcommand{\Cost}{\text{Cost}\xspace}

\renewcommand{\shortauthors}{Johnson et al.}

\author{Ezra Johnson}
\email{ezrajohnson@caltech.edu}
\affiliation{%
  \institution{California Institute of Technology}
  \city{Pasadena}
  \state{California}
  \country{USA}
}

\author{Adam Lechowicz}
\email{alechowicz@cs.umass.edu}
\affiliation{%
  \institution{University of Massachusetts Amherst}
  \city{Amherst}
  \state{Massachusetts}
  \country{USA}
}

\author{Mohammad Hajiesmaili}
\email{hajiesmaili@cs.umass.edu}
\affiliation{%
  \institution{University of Massachusetts Amherst}
  \city{Amherst}
  \state{Massachusetts}
  \country{USA}
}

\begin{abstract}

A wide range of sustainability and grid-integration strategies depend on workload shifting, which aligns the timing of energy consumption with external signals such as grid curtailment events, carbon intensity, or time-of-use electricity prices. The main challenge lies in the online nature of the problem: operators must make real-time decisions (e.g., whether to consume energy now) without knowledge of the future. While forecasts of signal values are typically available, prior work on learning-augmented online algorithms has relied almost exclusively on simple point forecasts. In parallel, the forecasting research has made significant progress in uncertainty quantification (UQ), which provides richer and more fine-grained predictive information.
In this paper, we study how online workload shifting can leverage UQ predictors to improve decision-making. We introduce \texttt{UQ-Advice}, a learning-augmented algorithm that systematically integrates UQ forecasts through a \textit{decision uncertainty score} that measures how forecast uncertainty affects optimal future decisions. By introducing \textit{UQ-robustness}, a new metric that characterizes how performance degrades with forecast uncertainty, we establish theoretical performance guarantees for \texttt{UQ-Advice}. Finally, using trace-driven experiments on carbon intensity and electricity price data, we demonstrate that \texttt{UQ-Advice} consistently outperforms robust baselines and existing learning-augmented methods that ignore uncertainty.

\end{abstract}

\maketitle

\section{Introduction}
\label{sec:intro}

The urgency of climate change and the growing need for deep integration between the power grid and energy-intensive infrastructures such as data centers have spurred growing interest in rethinking how these systems operate, with a focus on making them adaptive to external signals.
One of the most promising approaches to achieve such a goal is \textit{workload shifting}, where the timing of computational and energy demand is aligned with external signals like grid curtailment events (to take advantage of surplus energy supply), carbon intensity (to reduce carbon footprint), or time-of-use electricity prices (to reduce operational expenses).  A wide variety of energy-consuming tasks accommodate workload shifting, including batch jobs in data centers~\cite{radovanovic2022carbon, Wiesner:21, Kim:2023, Dodge:2022}, electric vehicle charging at home or work~\cite{Cheng:22:EV, Wang:25:EV}, and HVAC scheduling~\cite{Kurte:23:HVAC, Yang:25:HVAC}, among others.  While workload shifting exhibits great potential for savings in these applications \cite{sukprasert2023quantifying, Bovornkeeratiroj:23}, a key challenge is the online nature of the problem: scheduling decisions must be made sequentially, usually without knowledge of future external signal values.

In practice, reasonably accurate forecasts of the relevant signals can often be obtained. For example, recent advances have improved the forecasting of carbon intensity values~\cite{Maji:22:CC, Zhang:23, Yan:25}, and energy price forecasting has long been an established area of study~\cite{Amjady:06, Pindoriya:08, Patel:21}. The algorithmic literature of workload shifting also studied leveraging such predictions in decision making~\cite{Lechowicz:24, Hanafy:23:CarbonScaler}. This literature, however, almost exclusively focused on \textit{point predictions} that directly predict future signal values.
This limits their applicability since point predictions provide no systematic way to quantify uncertainty, leaving algorithms themselves to determine how much trust to place in forecasts, a challenge highlighted in the literature on algorithms with predictions~\cite{Sun:24}.
More recently, however, methods such as conformal prediction have emerged to provide high-quality uncertainty quantification (UQ) for these forecasts~\cite{Li:24, Alghumayjan:24}. Unlike point predictions, UQ methods produce not only a predicted value but also an associated confidence measure, often in the form of an interval (e.g., upper and lower bounds) that is guaranteed to contain the true signal value with a specified probability, e.g., 95\%.  Incorporating such fine-grained uncertainty information can improve algorithmic decision-making and remove the need for manual adjustments of the user-specified ``trust'' in machine-learned forecasts seen in prior work~\cite{Lechowicz:24}.

Given the availability of these richer, uncertainty-aware forecasts, a natural question is how this type of forecast could enable more robust and performant online decision-making for signal-aware workload shifting. In this paper, we follow the paradigm of learning-augmented online algorithms~\cite{Lykouris:18, Purohit:18}, which provide a framework to systematically incorporate and analyze the performance of algorithms that combine the ``best-of-both-worlds:'' a robust baseline with machine-learned advice.  A few recent works~\cite{Sun:24, Shen:25, Bostandoost:24, Daneshvaramoli:25, Zhao:17, Anand:22} have begun to consider incorporating uncertainty-quantified predictions in canonical online problems such as ski rental and online search, demonstrating the benefits of these richer prediction models. However, the application of these ideas to more complex, practical problems like workload shifting remains underexplored.  A key analytical challenge is that incorporating uncertainty-quantified predictors in algorithm design is already known to be challenging for simple stylized problems (e.g., ski rental~\cite{Sun:24}).  Extending these ideas to a setting such as workload shifting, where UQ predictors are multi-stage (i.e., multi-time-step) in nature and the underlying problem is more complex, is quite non-trivial.
Similarly, learning-augmented algorithms have been considered for online problems related to workload shifting, such as one-way trading~\cite{SunLee:21, Lee:24, Lechowicz:24, Christianson:22}, but these all consider point prediction models that are not \textit{uncertainty-aware}.

\smallskip
\noindent\textbf{Contributions.}
This paper addresses the problem of online workload shifting with uncertainty-quantified predictors through the lens of the learning-augmented online algorithms framework. We make three main contributions, each corresponding to a central research question.

\noindent$\blacktriangleright$ \textit{How can an online algorithm for workload shifting systematically incorporate UQ signal forecasts to balance trust in the prediction against worst-case performance automatically?}
To address this question, we investigate how UQ predictions can be effectively leveraged in workload shifting, modeling the problem as a variant of the online conversion with switching costs (\OCS) framework. \OCS captures key trade-offs in workload shifting applications such as batch job scheduling in data centers and electric vehicle charging. Building on this foundation, we propose \UQAdv, a novel learning-augmented online algorithm that pre-computes uncertainty-driven scenarios to determine when and how to rely on UQ forecasts. \UQAdv dynamically interpolates between a prediction-based solution and a classical competitive online algorithm. Our analysis develops a rich theoretical model of UQ forecasts, inspired by state-of-the-art empirical techniques such as conformal prediction, thereby tailoring the approach to practical forecasting methods.

\noindent$\blacktriangleright$ \textit{Can such an algorithm achieve theoretical bounds on its performance that gracefully degrade as the quantified uncertainty in the forecast increases?} 
We provide a rigorous theoretical analysis of our algorithm, establishing guarantees on its \textit{consistency} and \textit{robustness}. We also introduce and analyze a new performance metric, \textit{uncertainty-quantified robustness}, which characterizes the algorithm's performance when the uncertainty quantification is correct but forecasts are not exact.
Our result for this metric implies that \UQAdv is $1$-competitive (i.e., optimal) when uncertainty is low, with performance that smoothly degrades as a function of increasing uncertainty (see \autoref{thm:uq-uqrobustness} for a formal statement).
To our knowledge, our work is the first to consider complex multi-time-step UQ predictions in the design of online algorithms.

\noindent$\blacktriangleright$ \textit{How much and in which cases does an algorithm leveraging UQ signal forecasts improve workload shifting performance in practice?} Through extensive trace-driven experiments using carbon intensity and energy price as signals, we demonstrate that \UQAdv significantly improves on existing learning-augmented techniques that do not utilize information about uncertainty, improving on baselines by up to 12.6\% on average and up to 26.15\% in the 95$^\text{th}$ percentile, all while requiring less fine-tuning than existing approaches.  In this respect, we highlight \autoref{fig:uncertainty}, which shows that \UQAdv effectively ``matches'' the performance of the best algorithm in all scenarios -- relying on forecasts when uncertainty is low and matching a robust baseline when forecasts are poor.

\section{Problem and Preliminaries}
\label{sec:prob}

\noindent In this section, we formalize the online signal-aware workload shifting problem and present preliminaries and assumptions used in the paper.  Throughout the paper, lowercase bold letters denote vectors -- e.g., for a vector of length $T$ we have $\mathbf{x} \coloneqq \{ x_t \}_{t \in [T]}$.

\subsection{Problem Formulation}
We introduce the \textbf{s}ignal-\textbf{a}ware workload \textbf{s}hifting \textbf{p}roblem (\SASP), where the decision-maker's goal is to complete a workload of unit size (without loss of generality) before a given deadline $T$ while minimizing their total cost.  At each time step $t \in [T]$, a price signal $p_t$ arrives online, and the decision-maker must choose the amount of the workload to complete at the current time step, represented by $x_t$.  $x_t$ is $\geq 0$ and $\leq d_t$, where $d_t \leq 1$ represents a \textit{rate constraint} that potentially constrains the decision-maker from e.g., completing the entire workload in a single step.
Given a decision $x_t$, the decision-maker's per-step cost is $p_t x_t + \beta \vert x_t - x_{t-1} \vert + \lambda x_t^2$, where the first term is the cost due to the time-varying signal (e.g., carbon emissions, energy procurement cost), the second term penalizes the decision-maker for ``unsmooth jumps'' across time steps, and the third term is a (typically small) regularization term that encourages spreading resource usage across multiple time steps.  The coefficients $\beta > 0$ and $\lambda > 0$ are both known apriori.  An offline formulation of \SASP is as follows:
\begin{align}
    \SASP: \min_{\mathbf{x} \coloneqq \{x_t\}_{t\in T}} & \underbrace{\sum_{t=1}^{T} p_t x_t}_{\text{total signal cost}} + \underbrace{\sum_{t=1}^{T+1} \beta \vert x_t - x_{t-1} \vert,}_{\text{switching penalty}} + \underbrace{\lambda \sum_{t=1}^T x_t^2}_{\text{regularizer}} \label{eq:obj} \\
    \text{s.t.} & \underbrace{\sum_{t=1}^{T} x_t = 1,}_{\text{deadline constraint}} \quad  x_t \in [0, d_t] \ \forall t \in [T].
\end{align}

In this work, we focus on designing algorithms for the \textit{online} version of the \SASP problem, where the decision-maker can only observe price signals up to time $t$ when selecting $x_t$, and each choice of $x_t$ is irrevocable (i.e., it cannot be revised at future time steps). A closely related problem, called \textit{online conversion with switching costs} (\OCS), was proposed in~\cite{Lechowicz:24}, building on a long line of work in online search problems with similar structure~\cite{ElYaniv:01, Zhou:08}. Compared to \OCS, our online formulation of \SASP introduces two key extensions: (i) the addition of a regularization term $\lambda \sum_{t=1}^T x_t^2 = \lambda \Vert \mathbf{x} \Vert_2^2$, which facilitates analysis of the problem (see \autoref{sec:leverage})\footnote{Adding a small regularization term is a standard technique to ensure strong convexity of the objective in \eqref{eq:obj}, which we exploit in our theoretical results.}, and (ii) the incorporation of uncertainty-quantified forecasts, as detailed below.  

Although the problem is abstractly framed in terms of ``prices,'' the time-varying signal cost minimized in \SASP naturally models a broad range of applications. Examples include carbon emissions minimization (where the signal is time-varying grid carbon intensity) and electricity cost minimization (where the signal is time-of-use energy pricing), among others.  

In the UQ setting, we focus on a practically motivated scenario in which the decision-maker has advance access to \textit{uncertainty-quantified forecasts} of the price signal. In the following, we formally define the prediction model adopted in our analysis.  

\begin{definition}[The UQ forecast model]\label{dfn:uq-model}
The decision-maker is provided a \textit{point forecast} of the price signal over the next $T$ time slots, represented by $\{ \hat{p}_t \}_{t \in [T]}$, where an exactly correct forecast implies $\hat{p}_t = p_t$. 

In addition to this point forecast, the decision-maker also observes an \textit{uncertainty box set} for each time step, represented by $\{ \ell_t \}_{t \in [T]}$ and $\{ u_t \}_{t \in [T]}$, where $\ell_t \leq \hat{p}_t \leq u_t$.  Moreover, a coverage parameter $\delta$ guarantees that the true value $p_t$ lies within the interval $[\ell_t, u_t]$ with probability $1-\delta$.  The parameter $\delta$ is usually known and fixed apriori, and is sometimes referred to as the coverage~\cite{Alghumayjan:24}.  Intuitively, as the width of the interval $[\ell_t, u_t]$ increases, there is more uncertainty in the forecast's estimate.
\end{definition}

The prediction model defined in \sref{Definition}{dfn:uq-model} is primarily inspired by empirical work that proposes methods for time-series forecasting in relevant domains.
We particularly highlight that our model considers \textit{both} point forecasts and uncertainty-quantified sets for each signal value, which deviates from prior work on uncertainty-quantified predictions in online algorithms~\cite{Sun:24, Daneshvaramoli:25}, but aligns more closely with recent works proposing e.g., conformal prediction methods for carbon intensity forecasting and energy markets~\cite{Li:24, Alghumayjan:24}.  

\subsection{Performance Metrics}

We study \SASP within the framework of online algorithms, with a particular focus on learning-augmented algorithm design. In the classical setting of online algorithms, the objective is to design an algorithm that achieves a small \textit{competitive ratio}~\cite{Borodin:92}, defined as:

\begin{definition}[Competitive Ratio] 
Let $\OPT(\mathcal{I})$ denote the offline optimal solution on an arbitrary \SASP input $\mathcal{I}$ and $\ALG(\mathcal{I})$ represent the solution of an online algorithm \ALG on the same input.
Given a cost function $\text{Cost}(\cdot)$, we say that \ALG is \textbf{$\alpha$-competitive} iff the following holds: $\text{Cost}(\ALG(\mathcal{I})) \leq \alpha \text{Cost}(\OPT(\mathcal{I})) \ \forall \mathcal{I} \in \Omega$, where $\Omega$ is the universe of all possible \SASP inputs.  Note that $\alpha$ is greater than or equal to $1$, and the smaller it is, the closer the algorithm is to the optimal solution.  
\end{definition}
In the recent literature on learning-augmented algorithms~\cite{Lykouris:18, Purohit:18}, the competitive ratio metric is interpreted via the notions of consistency and robustness, defined as follows:
\begin{definition}[Consistency and Robustness]
A learning-augmented algorithm \ALG is $\eta$-consistent if it is \textbf{$\eta$-competitive} when it is given an accurate (exactly correct) prediction and \textbf{$\alpha$-robust} if it is $\alpha$-competitive regardless of the prediction's quality.
\end{definition}
These quantities respectively measure how close an algorithm's solution is to the offline optimal when the prediction is exactly correct, and bound how far the algorithm's solution strays from the offline optimal when the prediction is arbitrarily wrong.

In our setting, leveraging the expressiveness of UQ predictions (see \sref{Definition}{dfn:uq-model}), we additionally consider a new metric of \textit{UQ-robustness}. The purpose of introducing this model is to capture the case where the prediction is not exactly correct (e.g., the point forecasts do not match the true signal values) but the input instance is still covered by the uncertainty set (i.e., with probability $1-\delta$).  It is defined as follows:
\begin{definition}[UQ-Robustness]
A learning-augmented algorithm \ALG is \textbf{$\gamma$-UQ-robust} if, given any set of uncertainty-quantified forecasts, it is at most $\gamma$-competitive on the entire subset of \SASP inputs that lie \textbf{within} the uncertainty sets defined by $\{ \ell_t \}_{t \in [T]}$ and $\{ u_t \}_{t \in [T]}$.
\end{definition}

\subsection{Assumptions}
We place the following assumptions on the \SASP problem throughout the paper.
Following prior work, we assume that signal values have bounded support, i.e., $p_t \in [p_{\min}, p_{\max}] \ \forall t \in [T]$, where $p_{\min}$ and $p_{\max}$ are known positive constants.  This has been shown to be a necessary assumption for any competitive algorithm for the online conversion problem and other related online problems~\cite{ElYaniv:01, Lorenz:08, Lechowicz:24, SunZeynali:20}, and intuitively holds in practice since the signals of interest for workload shifting (e.g., price signals, carbon intensity) are bounded.

We also assume that the switching coefficient $\beta$ and regularization coefficient $\lambda$ are ``not too large''. Formally, both are bounded within $\beta \in \left[ 0, \frac{(p_{\max} - p_{\min})}{2} \right)$ and $\lambda \in [0, p_{\max} - p_{\min})$, respectively.  If either $\beta$ or $\lambda$ exceeds these ranges, their impact on the total cost exceeds that of the signal value, and any competitive algorithm should simply minimize these terms, making decision-making trivial~\cite{Lechowicz:24}.

We typically assume that the deadline $T$ is known in advance to facilitate both the uncertainty-quantified forecasts and a ``compulsory trade'' that ensures the entire workload is completed before the end of the sequence.   If the decision-maker has completed has completed $w_j$ fraction of the workload at time $j$, this compulsory trade begins whenever $\sum_{t=j+1}^T d_t < 1 - w^{(j)}$ (i.e., when the remaining time steps is barely enough to complete the workload). During this compulsory execution, a signal-agnostic algorithm takes over and makes maximal decisions $x_t = d_t$ for the remaining time steps.  If future values of $d_t$ are unknown, we assume the decision-maker is still given a signal to indicate that the deadline is coming up and that the compulsory trade should begin.

\section{Leveraging Uncertainty Estimates}
\label{sec:leverage}

\noindent In this section, we explore exactly how uncertainty-quantified forecasts can help inform decisions for the \OCU problem. We first note that while prior work has explored uncertainty-quantified predictions in the learning-augmented setting~\cite{Hajiesmaili:16, Daneshvaramoli:25, Sun:24}, the uncertainty has typically been limited to a single feature of the input.

For example, \citet{Sun:24} consider the ski rental problem, in which a decision-maker aims to ski for an unknown time horizon $N$. Each day the decision-maker must decide whether to rent skis at a cost of $\$1$ per day, or buy the skis at a cost of $\$B$ and ski for free from that point forwards. The goal is to minimize the cost of buying and renting skis.  In this problem, if the time horizon $N$ is known apriori, the optimal solution is simple -- always rent if $N < B$, and buy otherwise.  Therefore, in the \textit{uncertainty-quantified} regime, the decision-maker receives a prediction about the number of skiing days $N$, and a box uncertainty set (lower and upper bounds) for the true number of skiing days.

However, for more complicated multi-stage online problems, including canonical problems like metrical task systems and $k$-server, learning-augmented algorithms typically consider more expressive predictions~\cite{Antoniadis:20MTS, Christianson:23MTS, Anand:22, Antoniadis:23, Cho:22, Li:24:decentralized}.  For instance, learning-augmented algorithms for online conversion with switching costs use \textit{black-box decision advice}~\cite{Lechowicz:24}, where the decision-maker directly receives predictions of the optimal decision to make at each time step.  Given a forecast of the price signal, such predictions are easy to obtain by solving an offline version of the problem that assumes the forecasts are correct.
These more complicated prediction models have been shown to be necessary for these more expressive online problems, but they do not neatly fit into the aforementioned frameworks developed to use uncertainty-quantified predictions in simpler online problems.  In particular, the \OCU setting presents one key challenge: \textit{how should one relate uncertainty in the signal forecast with uncertainty in the optimal decisions}?

Our effort to address the above question will begin with inspiration from the literature on robust optimization~\cite{Gabrel:14:RO, Bertsimas:11:RO} to define a \textbf{decision uncertainty score}, which we will utilize in our eventual algorithm design (see \autoref{sec:alg}).
Robust optimization is a well-studied research topic that deals broadly with uncertain data in optimization problems~\cite{Gorissen:15:RO}.  In the language of \OCU, rather than searching for a single solution that works for known parameters (e.g., known future signal values), robust optimization techniques typically search for a solution that minimizes the worst-case cost across all scenarios in a given uncertainty set.  In \OCU, the decision-maker is given uncertainty-quantified forecasts that can be compactly represented by the following set:
\[
\mathcal{U} \coloneqq \{ \mathbf{z} \in \mathbb{R}^T: \ell_t \leq z_t \leq u_t \}.
\]
In words, $\mathcal{U}$ represents the set of all possible future signal values $\mathbf{z}$ that respect the uncertainty bounds $\{ \ell_t \}_{t \in [T]}$ and $\{ u_t \}_{t \in [T]}$.
Given the set $\mathcal{U}$, a robust optimization formulation of the \OCU problem would optimize the following:
\begin{align*}
    \min_{\mathbf{x} \coloneqq \{x_t\}_{t\in T}} \max_{\mathbf{z} \in \mathcal{U}} & \sum_{t=1}^{T} z_t x_t + \sum_{t=1}^{T+1} \beta \vert x_t - x_{t-1} \vert + \lambda \Vert \mathbf{x} \Vert_2^2, \\
    \text{s.t.} & \sum_{t=1}^{T} x_t = 1, \quad  x_t \in [0, d_t] \ \forall t \in [T].
\end{align*}
We note that the optimization is over both the decisions $\{x_t\}_{t\in[T]}$ and the ``worst-case sequence of prices'' $\mathbf{z}$.

While this ``classic'' robust optimization formulation of \OCU is instructive, the \textit{solutions} to this problem do not yet capture our desired connection between uncertainty in the signal and uncertainty in the decisions.  In fact, this robust optimization will reliably find that the ``worst-case sequence of prices'' is simply the case when $\{z_t = u_t \}_{t\in [T]}$ (as this sequence maximizes the overall cost), and the optimal decisions will reflect this -- from these outcomes, we cannot rigorously say anything about how much variability there is in the optimal \textit{decisions}.

To address this, we briefly review some notation and present our decision uncertainty score, which takes elements of the classic robust optimization model and reorients them to directly quantify the uncertainty in the optimal decisions due to uncertainty in the forecast.
Let $\hat{\mathbf{p}}$ denote the point forecasts provided in \sref{Def.}{dfn:uq-model}.  Slightly abusing notation, we let $\OPT(\mathbf{z}) \in [0,1]^T$ denote the optimal offline \OCU decisions given a sequence of prices $\mathbf{z}$.  
We are particularly interested in comparing the optimal decisions that assume the point forecasts are correct (i.e., $\OPT(\hat{\mathbf{p}})$ against the optimal decisions for any other sequence of prices within the uncertainty set.  To this end, we define the following \textbf{d}ecision \textbf{u}ncertainty \textbf{s}core (\DUS):
\begin{align}
    \DUS(\mathcal{U}, \hat{\mathbf{p}}) \coloneqq \max_{\mathbf{z} \in \mathcal{U}} \Vert \OPT(\hat{\mathbf{p}}) - \OPT(\mathbf{z})\Vert_1 \label{eq:dus}
\end{align}
Intuitively this quantity describes, in absolute terms, how much the optimal decision could change between the point forecasts ($\hat{\mathbf{p}}$) and \textit{any other} scenario that lies within the uncertainty set ($\mathbf{z} \in \mathcal{U}$).  By definition of the feasible decisions for the \OCU problem, the $\DUS(\mathcal{U}, \hat{\mathbf{p}})$ for any $\mathcal{U}$ and $\hat{\mathbf{p}}$ is strictly bounded between $0$ and $2$ -- this is the case because the $\ell_1$ norm of any solution to the \OCU problem is always $1$, and the difference between them is then at most $2$.

In our algorithm, we use the decision uncertainty score as a measure of \textit{trust} in the uncertainty-quantified forecasts.  Unfortunately, while the formulation of \DUS is clean, the optimization problem defined by \eqref{eq:dus} is non-convex --  in particular, while $\OPT(\mathbf{z})$ is the solution to a convex problem, the function $f(\mathbf{z}) = \Vert \OPT(\hat{\mathbf{p}}) - \OPT(\mathbf{z})\Vert_1$ that we are maximizing is generally non-concave, making the overall problem non-convex.

Despite the optimization problem's non-convexity (i.e., difficulty), the structure of the \OCU problem and the objective in \eqref{eq:dus} does allow us to prove something about the global optimal solution.  In particular, we have the following:

\begin{theorem} \label{thm:global-optimum-dus}
Let $y^\star$ denote the true global optimal value of $\DUS(\mathcal{U}, \hat{\mathbf{p}})$ (in \eqref{eq:dus}).
A global non-convex optimizer can find a solution $\hat{y}$ that satisfies the following inequality:
\[
y^\star \leq \hat{y} + \varepsilon,
\]
for any fixed $\varepsilon > 0$ in finite time, using $O\left( \left( \frac{\sqrt{T} \cdot \text{diam}(\mathcal{U})}{2 \lambda \cdot \varepsilon} \right)^T \right)$ iterations~\cite{Malherbe:17,  Bachoe:21}.
\end{theorem}

We defer the proof of \autoref{thm:global-optimum-dus} to \sref{Appendix}{apx:proof-leverage}.  The key idea in the proof is to show that the function $f(\mathbf{z}) = \Vert \OPT(\hat{\mathbf{p}}) - \OPT(\mathbf{z})\Vert_1$ is Lipschitz continuous in $\mathbf{z}$, which allows us to apply standard results from optimization literature~\cite{Malherbe:17,  Bachoe:21}.

Beyond this theoretical guarantee, we have found that it is practical to solve this problem in practice using solvers such as Gurobi or NCVX~\cite{gurobi, liang2022ncvx, curtis2017bfgssqp}.
To motivate this and demonstrate the intuitive meaning of the decision uncertainty score, \autoref{fig:dus_opt} shows an example of solving \eqref{eq:dus} for two instances of the \ROCU problem with $T=10$ time steps. 

In these two instances, the point predictions are identical.  However, the uncertainty sets differ -- in the left instance, the uncertainty set is small (i.e., the forecasts are very certain), while in the right instance, the uncertainty set is large (i.e., the forecasts are very uncertain).  In both cases, we solve for the decision uncertainty score and plot two decision sequences: $\OPT(\hat{\mathbf{p}})$ (the optimal decisions that assume the point forecasts are correct), and $\OPT(\mathbf{z})$ (the optimal decisions for the worst-case $\mathbf{z} \in \mathcal{U}$ defined by \eqref{eq:dus}) In the left instance, the decision uncertainty score is small ($\DUS(\mathcal{U}_1, \hat{\mathbf{p}}) = 1.36e^{-11}$), and the optimal decisions for the worst-case $\mathbf{z}$ are very similar to those for the point forecasts.  In contrast, in the right instance, the decision uncertainty score is large ($\DUS(\mathcal{U}_2, \hat{\mathbf{p}}) = 2$), and the optimal decisions for the worst-case $\mathbf{z}$ differ significantly from those for the point forecasts.  This example illustrates how the decision uncertainty score translates uncertainty in the price signal forecasts into a measure of the amount of variability in the possible optimal decisions.

\begin{figure}[h]
    \centering
    \includegraphics[width=0.95\linewidth]{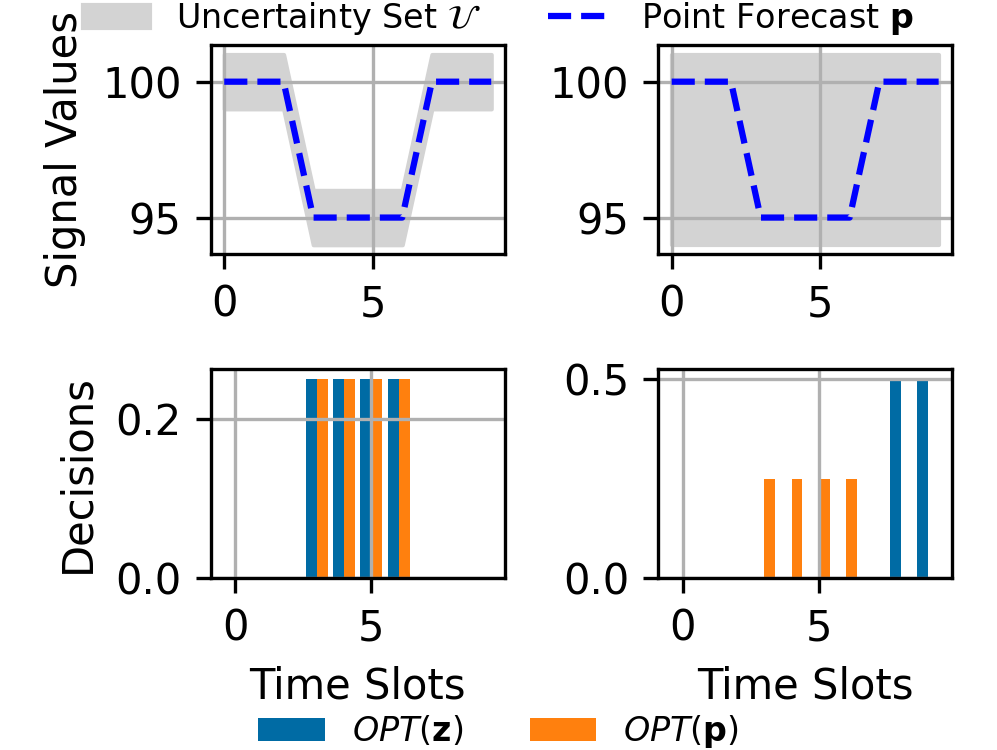}
    \caption{Two instances of \ROCU with identical point forecasts but different uncertainty sets to illustrate our \textit{decision uncertainty score} (\DUS).  The left instance has a small uncertainty set and a small \DUS, while the right instance has a large uncertainty set and a large \DUS.  In both instances, we plot the optimal decisions for the point forecasts (orange) and the optimal decisions for the worst-case scenario in the uncertainty set (blue).  The \DUS score captures how much the optimal decisions can vary due to uncertainty in the forecasts.}
    \label{fig:dus_opt}
\end{figure}

For the instances illustrated in \autoref{fig:dus_opt}, solving for the decision uncertainty score takes an average of 0.07964 seconds on a MacBook Pro M2 with 16 GB RAM using NCVX.  This suggests that it is practical to compute the decision uncertainty score in real time for most applications of interest.

\section{Algorithm Design}
\label{sec:alg}

In this section, we give a brief background on existing online approaches to solve the workload shifting problem before presenting our main algorithm, \UQAdv, which leverages the decision uncertainty score (\DUS) defined in the previous section to inform its decision-making process.

\subsection{Background}

In the classic competitive setting (i.e., without forecasts), the (not regularized) \OCS problem is solved optimally (i.e., it achieves the best competitive ratio possible) by an online algorithm framework called ``ramp-on, ramp-off'' (\RORO), shown by \cite{Lechowicz:24}.  In what follows, we present this algorithm and prove its competitive ratio when applied to the \ROCU problem.

In \RORO~\cite{Lechowicz:24}, the online decision at each time step is made by solving a \emph{pseudo-cost minimization problem} to determine the amount of workload to complete at the current time step (i.e., $x_t \in [0, \min(1-w_{t-1}, d_t)]$).  This minimization balances between the extreme options of completing ``too much'' of the workload early (thus incurring suboptimal costs if prices later drop) and waiting too long to complete the workload (thus risking being forced to complete a large portion of the workload at once, potentially at high cost).

Whenever the price signal is ``sufficiently attractive'', the pseudo-cost minimization finds the best decision that completes just enough of the workload to maintain a certain competitive ratio.  To rigourously define this trade-off, the authors in \cite{Lechowicz:24} introduce a \emph{dynamic threshold function} $\phi(w) : [0,1] \rightarrow [p_{\min},p_{\max}]$, which is a function that maps the current utilization $w$ to a price value.  It is defined as:
\begin{equation}
    \phi(w) = p_{\max} - \beta + \left( \frac{p_{\max}}{\alpha_{\RORO}} - p_{\max} + 2 \beta \right) \exp \left( \frac{w}{\alpha_{\RORO}} \right), \label{eq:phi}
\end{equation}
where $\alpha_{\RORO}$ is the optimal competitive ratio for \OCS, defined as~\cite{Lechowicz:24}:
{\small
\begin{equation}
    \alpha_{\RORO} \coloneqq \left[ W \left( \left( \frac{2\beta + p_{\min}}{p_{\max}} - 1\right) \exp \left( \frac{2\beta}{p_{\max}} - 1 \right) \right) - \frac{2\beta}{p_{\max}} + 1 \right]^{-1}. \label{eq:alpha_roro}
\end{equation}
}

\noindent In the above, $W(\cdot)$ is the Lambert $W$ function~\cite{Corless:96LambertW}.  Given this definition of $\alpha_{\RORO}$, note that $\phi(\cdot)$ is monotonically decreasing on the interval $w \in [0,1]$.  We summarize the \RORO algorithm in \sref{Algorithm}{alg:roro}.
\begin{algorithm}[t]
\caption{Online Ramp-On, Ramp-Off Algorithm (\RORO)~\cite{Lechowicz:24}}\label{alg:roro}
{\small
    \begin{algorithmic}[1]
        \item \textbf{input: } threshold function $\phi(w) : [0,1] \rightarrow [p_{\min},p_{\max}]$
        \item \textbf{initialize: } initial decision $x_0 = 0$, current utilization $w_0 = 0$
        \While {price $p_t$ is revealed and $w_{t-1} < 1$}:
            \State{solve \textbf{pseudo-cost minimization problem} to obtain decision:}
            \begin{align*}
                x_t = \kern-2em \argmin_{x \in [0,\min( 1-w_{t-1}, d_t )]} \kern-1em & p_t x_t + \beta \vert x - x_{t-1} \vert - \int_{w_{t-1}}^{w_{t-1}+x} \phi(u) du
            \end{align*}
            \State{update utilization as $w^{(t)} = w_{t-1} + x_t$}
        \EndWhile
    \end{algorithmic}
}
\end{algorithm}
For the \ROCU problem, we can prove that \RORO achieves almost the same competitive ratio as in \OCS, with an additive term that depends on the regularization.  Specifically, we have the following theorem:
\begin{theorem}
\label{thm:roro-comp-rocu}
    The \RORO algorithm given by \cite{Lechowicz:24} and summarized in \sref{Algorithm}{alg:roro} is $\alpha$-competitive for \ROCU, where $\alpha$ is defined as:
    \begin{equation}
        \alpha \coloneqq \frac{T(\alpha_{\RORO} \cdot p_{\min} + \lambda)}{T p_{\min} + \lambda},
    \end{equation}
    where $\alpha_{\RORO}$ is defined in \eqref{eq:alpha_roro}, $T$ is the time horizon, $p_{\min}$ is the minimum price, and $\lambda$ is the regularization parameter.  Note that as $\lambda \rightarrow 0$, we have $\alpha \rightarrow \alpha_{\RORO}$.
\end{theorem}

\noindent We defer the proof of \autoref{thm:roro-comp-rocu} to \sref{Appendix}{apx:proof-roro}.  Building on top of the \RORO algorithm, the authors in \cite{Lechowicz:24} also propose a learning-augmented algorithm called \ROAdv, which leverages \textit{black-box advice} (i.e., direct predictions of the optimal decisions) to improve average-case performance. 
The key idea in \ROAdv is to make decisions that are a convex combination of the online decisions given by \RORO and the advice at each time step -- this is a relatively standard technique in the literature on learning-augmented algorithms (see, e.g., \cite{Antoniadis:20MTS}).  
However, an important detail of \ROAdv is that it requires a \emph{trust parameter} $\lambda \in [0,1]$ to be set by the user in advance -- this parameterizes how much \ROAdv ``trusts'' the advice versus the robust online decisions.  In practice, this parameter must be fine-tuned to achieve good performance, which can be challenging if the quality of the advice is unknown or varies over time.

In our approach (described below), we attempt to leverage the extra information provided by uncertainty-quantified forecasts to automatically and adaptively determine how close decisions should be to the advice versus the robust online decisions, without requiring a trust parameter to be set by the user in advance.

\subsection{The \UQAdv algorithm}

In this section, we present our main algorithm, \UQAdv, which takes uncertainty-quantified forecasts (as defined in \sref{Def.}{dfn:uq-model}) of the relevant signal as input.

The key idea in \UQAdv is to precompute several relevant quantities based on the uncertainty-quantified forecasts, and then use these to determine how much to trust the forecasts.  Specifically, \UQAdv first computes $\{\hat{x}_t\}_{t=1}^T \coloneqq \OPT(\hat{\mathbf{p}})$, the optimal decisions that assume the point predictions are correct $\{\hat{p}_t\}_{t=1}^T$ (i.e., ignoring uncertainty).  It also computes the decision uncertainty score $\DUS(\mathcal{U}, \hat{\mathbf{p}})$ as defined in \eqref{eq:dus}, which quantifies the variability in the (true) optimal decisions based on the point predictions.  

Based on the value of $\DUS(\mathcal{U}, \hat{\mathbf{p}})$, \UQAdv then computes a mixing parameter $\gamma \in [0,1]$ as:
\begin{equation}
    \gamma \gets 1 - \frac{\DUS(\mathcal{U}, \hat{\mathbf{p}})}{2},
\end{equation}
where note that since the decision uncertainty score is bounded in the interval $[0,2]$, we have $\gamma \in [0,1]$.
Intuitively, $\gamma$ quantifies how much the optimal decisions vary based on the uncertainty-quantified forecasts -- if $\gamma$ is large (i.e., close to 1), then the optimal decisions are relatively insensitive to the uncertainty, and thus \UQAdv can ``trust'' the point predictions more.  Conversely, if $\gamma$ is small (i.e., close to 0), then the optimal decisions vary significantly based on the uncertainty, and thus \UQAdv should ``trust'' the robust online decisions more.  We formally describe the \UQAdv algorithm in \sref{Algorithm}{alg:uq}.

\begin{algorithm}[t]
\caption{\UQAdv Algorithm}\label{alg:uq}
{\small
    \begin{algorithmic}[1]
        \item \textbf{input:} uncertainty-quantified predictions $\hat{\mathbf{p}} = \{\hat{p}_t\}_{t=1}^T$ with uncertainty set $\mathcal{U} = \{\ell_t, u_t\}_{t=1}^T$, threshold function $\phi(w) : [0,1] \rightarrow [p_{\min},p_{\max}]$. \vspace{0.1cm}
        \item \textbf{initialize:} initial decision $x_0 = 0$, current utilization $w_0 = 0$ \vspace{0.1cm}
        \State{compute advice $\hat{\mathbf{x}} \gets \OPT(\hat{\mathbf{p}})$}
        \State{solve for decision uncertainty score $\DUS(\mathcal{U}, \hat{\mathbf{p}})$ using \eqref{eq:dus}}
        \State{set mixing parameter $\gamma \gets 1 - \nicefrac{\DUS(\mathcal{U}, \hat{\mathbf{p}})}{2}$} \vspace{0.1cm}
        
        \While {price $p_t$ is revealed and $w_{t-1} < 1$}:
        \State{solve \textbf{pseudo-cost minimization} to obtain robust decision:}
            \begin{align*}
                \tilde{x}_t = \kern-2em \argmin_{x \in [0,\min( 1-w_{t-1}, d_t )]} \kern-1em & p_t x_t + \beta \vert x - x_{t-1} \vert - \int_{w_{t-1}}^{w_{t-1}+x} \phi(u) du
            \end{align*}
        \State{set online decision $x_t = \gamma \cdot \hat{x}_t + (1-\gamma) \cdot \tilde{x}_t$}
        \State{update utilization as $w^{(t)} = w_{t-1} + x_t$}
        
        \EndWhile
    \end{algorithmic}
}
\end{algorithm}

In what follows, we analyze the theoretical performance of \UQAdv in terms of three metrics: (i) \emph{consistency}, which quantifies performance when the point predictions are perfect; (ii) \emph{robustness}, which quantifies performance when the point predictions are adversarially bad; and (iii) a new metric we call \emph{uncertainty-quantified robustness}, which quantifies performance when the point predictions are imperfect but the uncertainty quantification is accurate.  We summarize our main theoretical results below, and defer the proofs to \autoref{apx:proof-alg}.

\begin{theorem}[Consistency of \UQAdv] \label{thm:uq-consistency}
    Under the assumption that the point predictions are perfect (i.e., $\hat{\mathbf{p}} = \mathbf{p}$), \UQAdv is $\eta$-consistent, where:
    \begin{equation*}
        \eta \coloneqq 1 + (\alpha - 1) \frac{\DUS(\mathcal{U}, \mathbf{\hat{p}})}{2} ,
    \end{equation*}
    where $\alpha$ is the competitive ratio of \RORO for \ROCU given in \autoref{thm:roro-comp-rocu}.  Note that as $\DUS(\mathcal{U}, \mathbf{\hat{p}}) \rightarrow 0$, we have $\eta \rightarrow 1$.
\end{theorem}

We defer the proof of this theorem to \sref{Appendix}{apx:proof-consistency}. Intuitively, the proof of \autoref{thm:uq-consistency} follows from the fact that when the point predictions are perfect, the advice $\hat{\mathbf{x}} = \OPT(\mathbf{\hat{p}})$ is optimal, and thus the performance of \UQAdv is a convex combination of the optimal performance and the performance of \RORO, weighted by $\gamma$ and $1-\gamma$, respectively.  

We note one feature of this result that arises due to UQ considerations -- if the point predictions are exactly correct but the decision uncertainty score is high (i.e., $\DUS(\mathcal{U}, \mathbf{\hat{p}}) \to 2$), $\eta$ approaches $\alpha$.  This mirrors how \UQAdv essentially ``hedges'' in a high-uncertainty scenario, choosing to more closely follow the robust \RORO algorithm to hedge against the possibility that the forecasts are not correct.

Replacing $\gamma \coloneqq \nicefrac{\DUS(\mathcal{U}, \mathbf{\hat{p}})}{2}$ with a manually specified trust parameter, this result exactly matches the consistency result of \ROAdv for the \OCS problem~\cite{Lechowicz:24}.

\noindent Next, we consider the robustness of \UQAdv:

\begin{theorem}[Robustness of \UQAdv] \label{thm:uq-robustness}
    Under the assumption that \textit{both} the point predictions and the uncertainty quantification are completely (i.e., adversarially) incorrect, \UQAdv is $\zeta$-robust, where:
    \begin{equation*}
        \zeta \coloneqq \left( 1 - \frac{\DUS(\mathcal{U}, \hat{\mathbf{p}})}{2} \right) \cdot \frac{p_{\max} + 2\beta + \lambda}{p_{\min} + \frac{\lambda}{T}} + \left( \frac{\DUS(\mathcal{U}, \hat{\mathbf{p}})}{2} \right) \cdot \alpha,
    \end{equation*}
    where $\alpha$ is the competitive ratio of \RORO for \ROCU given in \autoref{thm:roro-comp-rocu}.  Note that as $\DUS(\mathcal{U}, \mathbf{\hat{p}}) \rightarrow 2$, we have $\zeta \rightarrow \alpha$.
\end{theorem}

We defer the proof of this theorem to \sref{Appendix}{apx:proof-robustness}.  The proof follows by considering a worst-case scenario where each element of the uncertainty-quantified forecasts is arbitrarily bad, and showing that in this case, \UQAdv is upper bounded by a combination of the performance of \RORO (the $\alpha$ term) and the performance of the advice (which can be arbitrarily bad).  Replacing $\gamma \coloneqq \nicefrac{\DUS(\mathcal{U}, \mathbf{\hat{p}})}{2}$ with a manually specified trust parameter, this result approximately matches the robustness result of \ROAdv for the \OCS problem~\cite{Lechowicz:24}, with an additional additive term due to the regularization term not present in the \OCS problem.

We note one feature of the robustness case that arises from the decision uncertainty score -- if the point predictions are incorrect \textit{and} the decision uncertainty score is high (i.e., $\DUS(\mathcal{U}, \mathbf{\hat{p}}) \to 2$), $\zeta$ approaches $\alpha$, which is the existing competitive ratio without forecasts (the best we can hope to do if forecasts are not useful).  
This mirrors how \UQAdv essentially ``hedges'' in a high-uncertainty scenario, choosing to more closely follow the robust \RORO algorithm to hedge against the possibility that the forecasts are not correct.

On the other hand, if point predictions are incorrect but the decision uncertainty score is still low (i.e., $\DUS(\mathcal{U}, \mathbf{\hat{p}}) \to 0$), the UQ forecasts are incorrect as a whole and $\zeta$ approaches $\frac{p_{\max} + 2\beta + \lambda}{p_{\min} + \frac{\lambda}{T}}$.  Assuming the UQ forecasts are sufficiently well-calibrated, this should happen very rarely (i.e., with less than 5\% probability) -- below, we consider UQ-robustness, which captures the (more probable) case where point predictions are imperfect but the uncertainty quantification is accurate.  

\begin{theorem}[UQ-Robustness of \UQAdv] \label{thm:uq-uqrobustness}
    Under the assumption that the uncertainty quantification is accurate (i.e., the true signals lie within the uncertainty intervals), but the point predictions may be imperfect, \UQAdv is $\theta$-UQ-robust, where:

    {\small
    \begin{equation*}
        \theta \coloneqq 1 + \left( \frac{\DUS(\mathcal{U}, \hat{\mathbf{p}})}{2} \right) \left( \alpha - 1 + \left( 1 - \frac{\DUS(\mathcal{U}, \hat{\mathbf{p}})}{2}\right) \frac{p_{\max} - p_{\min} + 4\beta + 2\lambda}{p_{\min} + \frac{\lambda}{T}} \right),
    \end{equation*}
    }

    \noindent where $\alpha$ is the competitive ratio of \RORO for \ROCU given in \autoref{thm:roro-comp-rocu}.  
\end{theorem}

We defer the proof of this theorem to \sref{Appendix}{apx:proof-uqrobustness}.  The proof follows by combining the arguments used in the proofs of \autoref{thm:uq-consistency} and \autoref{thm:uq-robustness}, and leveraging the fact that when the uncertainty quantification is accurate, we can bound the performance of the advice $\ADV = \hat{\mathbf{x}}$ in terms of the performance of $\OPT(\mathcal{I})$ and the decision uncertainty score $\DUS(\mathcal{U}, \hat{\mathbf{p}})$.

This result effectively ``bridges'' the traditional learning-augmented notions of consistency and robustness by leveraging information provided by the decision uncertainty score $\DUS(\mathcal{U}, \hat{\mathbf{p}})$.  Considering the extreme cases, if forecasts are very good and $\DUS(\mathcal{U}, \mathbf{\hat{p}}) \rightarrow 0$, we have $\theta \rightarrow 1$ at the same time as the consistency result $\eta$ also approaches $1$.  On the other hand, if forecasts are highly uncertain and $\DUS(\mathcal{U}, \mathbf{\hat{p}}) \rightarrow 2$, we have $\theta \rightarrow \alpha$ at the same time as the robustness result $\zeta$ also approaches $\alpha$.  In this sense, the additional expressiveness of UQ forecasts allows us to obtain the ``best-of-both-worlds'' trade-off, achieving performance equal to $\OPT$ when forecasts are good while retaining performance comparable to the optimal algorithm without predictions (\RORO) when forecasts are uncertain.

\section{Evaluation}
\label{sec:exp}

In this section, we experimentally evaluate the performance of \UQAdv using synthetic instances of \ROCU based on carbon intensity and electricity price traces.

\subsection{Experimental Setup} \label{sec:exp-setup}
We first detail our setup for synthetic experiments.  We consider unit-size workloads with no rate constraints (i.e., $d_t = 1 \ \forall t \in [T]$).  Note that the unit-size assumption is without loss of generality since we consider performance relative to baselines that complete the same size workload.

\smallskip\noindent\textbf{Carbon intensity traces and forecasts. \ } We use historical grid carbon intensity data (reported in grams of CO$_2$ per kilowatt-hour) from Electricity Maps~\cite{electricity-map} for three grid regions in the U.S., namely \texttt{CAISO} (California ISO), \texttt{ERCOT} (ERCOT, Texas), and \texttt{ISONE} (ISO New England).  See \sref{Table}{tab:characteristics} for a summary of the data.  For these three grid regions, we obtain real uncertainty-quantified forecasts from~\citet{Li:24}.  
In our experiments that use these traces, the objective is to minimize the total carbon footprint, subject to a constraint that the total workload should be completed before the deadline $T$.  For instance, using a toy-example with $T = 2$, if $p_1 = 800$ and $p_2 = 200$, a schedule with $\mathbf{x} = [1, 0]$ incurs a carbon footprint of $800$, while a schedule with $\mathbf{x} = [0.5, 0.5]$ incurs a carbon footprint of $500$.

\smallskip\noindent\textbf{Energy price traces and forecasts. \ } We use historical zonal price data for New York City, U.S. during late 2022 through 2023.  Prices are reported in USD per megawatt-hour (MWh) -- see \sref{Table}{tab:characteristics} for a summary of the data.  For this price data, we obtain conformal uncertainty quantified forecasts from~\citet{Alghumayjan:24}.  
In our experiments that use these traces, the objective is to minimize the total procurement cost of energy, subject to a constraint that the total amount of energy should be purchased before the deadline $T$.  For instance, using a toy-example with $T = 2$, if $p_1 = \$25$/MWh and $p_2 = \$125$/MWh, a schedule with $\mathbf{x} = [1, 0]$ incurs a cost of $\$25$, while a schedule with $\mathbf{x} = [0.5, 0.5]$ incurs a cost of $\$75$. 

\begin{table}[h]
\centering %
\caption{Statistics for signal traces, including the duration, granularity, minimum value, and maximum value~\cite{Li:24, Alghumayjan:24}.} \vspace{-1em}
\label{tab:characteristics}
{\small
\begin{tabular}{|l|l|l|l|l|}
\hline
Trace  & Signal Type                       & Duration                                                                            & Min. & Max. \\ \hline
\texttt{CAISO}  & \multirow{3}{*}{Carbon Intensity} & \multirow{3}{*}{\begin{tabular}[c]{@{}l@{}} 7/2/22-12/27/22 \\ Hourly\end{tabular}} & 59.33     & 338.63    \\ \cline{1-1} \cline{4-5} 
\texttt{ERCOT}  &                                   &                                                                                     & 131.55    & 429.92    \\ \cline{1-1} \cline{4-5} 
\texttt{ISO-NE} &                                   &                                                                                     & 137.68    & 308.39    \\ \hline
\texttt{NY-ISO} & Electricity Price                 & \begin{tabular}[c]{@{}l@{}}12/15/22-12/28/23\\ Hourly\end{tabular}                & -8.10   & 3121.07    \\ \hline
\end{tabular}
}
\vspace{-1em}
\end{table}

\smallskip\noindent\textbf{Parameter settings. \ }
In the main experiment setting, we evaluate algorithms on 1,000 random instances of \ROCU. 
For experiments where we vary specific parameters (see \autoref{fig:varied}), we evaluate on 200 random instances per value of variation that we test.
For the time horizon $T$, we consider values in the range of 2 to 24 -- we typically set $T=8$, unless otherwise specified.
For the switching cost coefficient $\beta$, we typically fix $\beta = 20$, but also consider a range of $\beta \in [0, (U-L)/2]$ in a specific experiment.  We set the regularization parameter $\lambda = 0$.

In one experiment, we directly manipulate the \textit{quality} of the uncertainty sets provided to \UQAdv to understand the impact of varying uncertainty in both the UQ intervals and the point forecasts. 
We define and specify a parameter $\xi \in [0,1]$ -- based on $\xi$, we set a \textit{width} for the uncertainty sets as $\texttt{width} = \xi \cdot \frac{(p_{\max}-p_{\min})}{2}$.  Using this width parameter, we generate synthetic uncertainty forecasts as follows: given a single time-step's true signal value $p$, its uncertainty set is defined by a lower bound $\ell = p - \text{Unif}(0,1)\cdot \texttt{width}$, and an upper bound $u = \ell + \texttt{width}$, where $\text{Unif}(0,1)$ is a uniform random number between $0$ and $1$.   If $\ell$ falls below $p_{\min}$, it is first truncated to $p_{\min}$ before computing $u$.

Given these synthetic uncertainty-quantified forecasts, we use the solution of the decision uncertainty score (see \autoref{sec:leverage}) to find the signal sequence $\mathbf{z}$ that causes $\OPT(\mathbf{z})$ to deviate maximally from the true optimal decisions, setting the point forecasts $\hat{\mathbf{p}} = \mathbf{z}$.  This construction simulates a worst-case scenario that simultaneously increases uncertainty and generates less reliable predictions.
Note that when $\xi \to 0$, the point forecasts converge to the ground truth sequence, and when $\xi \to 1$, the forecasts significantly degrade.

\smallskip\noindent\textbf{Baselines. \ }
To evaluate \UQAdv, we compare it against several \textit{baseline} techniques.  First, we use CVXPY~\cite{CVXPY} to solve for the optimal offline solution $\OPT$ (i.e., given perfect knowledge of the true future signal values).
Representing the optimal baseline without predictions, we implement the \RORO algorithm discussed in \autoref{sec:alg}.  We also implement a simple baseline called \texttt{Threshold} that has been considered by prior work~\cite{Wiesner:21, Lechowicz:24} -- this technique simply chooses to allocate resources to the workload whenever the signal value is below $\sqrt{p_{\min} p_{\max}}$.

Finally, we compare \UQAdv against two baseline techniques that \textit{do} use predictions, namely point forecasts of the future signal values.  We implement the \ROAdv algorithm from prior work~\cite{Lechowicz:24} that combines a robust algorithm (i.e., \RORO) with black-box advice.  We initialize \ROAdv with two values for the trust parameter $\lambda$ (also sometimes called a ``combination factor'').  We consider $\lambda = 0.5$ as a simple baseline. On the traces with \textit{real uncertainty sets}, we find the optimal $\lambda$ that yields the best performance on the tested signal traces, calling this $\lambda^\star$.

In all experiments, we primarily report the empirical competitive ratio of different algorithms as the performance metric -- it is defined as $\nicefrac{\Cost(\ALG(\mathcal{I}))}{\Cost(\OPT(\mathcal{I}))} \geq 1$ for each instance $\mathcal{I}$.  Note that a smaller empirical competitive ratio is better.

\begin{table}[t]
\caption{Overall performance of \UQAdv compared to baseline techniques, aggregated over all experiments using real signal traces and forecasts.  Note that $\ROAdv[\lambda = \lambda^\star]$ indicates the baseline with the best ``trust parameter'' \textit{in hindsight} (i.e., if all instances are known in advance, see \autoref{sec:exp-setup} for details)
} \vspace{-1em}
\label{tab:results}
\begin{tabular}{|l|ll|}
\hline
\multirow{2}{*}{Algorithm} & \multicolumn{2}{l|}{Empirical Competitive Ratio} \\ \cline{2-3} 
 & \multicolumn{1}{l|}{average}  & 95th percentile  \\ \hline
\textbf{\UQAdv} & \multicolumn{1}{l|}{ \textbf{1.073882} }   & \textbf{1.311226}                \\ \hline
$\ROAdv[\lambda=\lambda^\star]$ & \multicolumn{1}{l|}{1.024320}        & 1.190832             \\ \hline
$\ROAdv[\lambda = 0.5]$    & \multicolumn{1}{l|}{ 1.189204 }        &  1.704182 \\ \hline
\RORO & \multicolumn{1}{l|}{ 1.363991 }        &  2.440307              \\ \hline
Simple Threshold & \multicolumn{1}{l|}{ 1.428846 }        & 2.407804                \\ \hline
\end{tabular}
\end{table}

\begin{figure}[t]
\centering
\begin{subfigure}{\columnwidth}
        \centering
        \includegraphics[width=0.9\columnwidth]{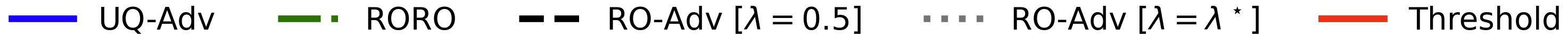}
        \label{fig:legend}
    \end{subfigure}
    \vspace{0.0cm}
    \begin{subfigure}{.48\columnwidth}
        \centering
        \includegraphics[width=\columnwidth]{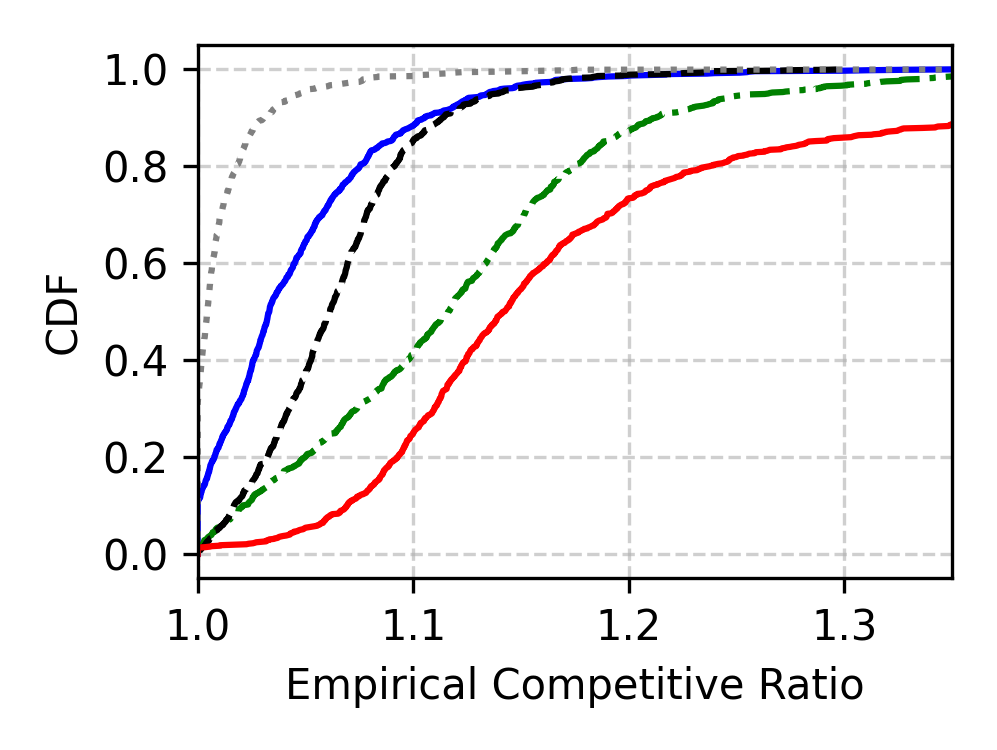}
        \caption{Carbon footprint}
        \label{fig:carbon_cdf}
    \end{subfigure}
    \hfill
    \begin{subfigure}{.48\columnwidth}
        \centering
        \includegraphics[width=\columnwidth]{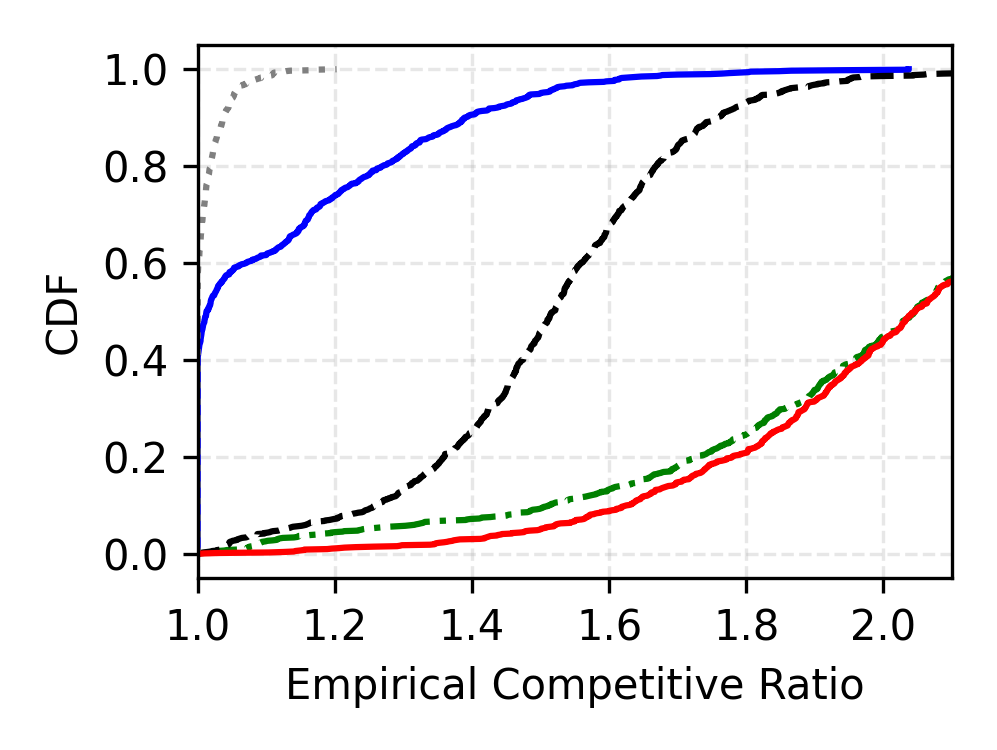}
        \caption{Electricity cost}
        \label{fig:elec_cdf}
    \end{subfigure}
    \caption{Cumulative distribution functions (CDFs) of the empirical competitive ratio for each tested algorithm in two experimental settings: (a) details those experiments where the objective is to minimize carbon footprint (i.e., carbon intensity signal traces), while (b) details experiments where the objective is to minimize the cost of electricity. }
    \label{fig:cdfs}
\end{figure}

\section{Experimental Results}

\begin{figure*}[t]
    \centering
    
    \begin{subfigure}{\textwidth}
        \centering
        \includegraphics[width=0.6\textwidth]{figures/legend.png}
        \label{fig:top}
    \end{subfigure}
        
    \begin{subfigure}{0.24\textwidth}
        \centering
        \includegraphics[width=\textwidth]{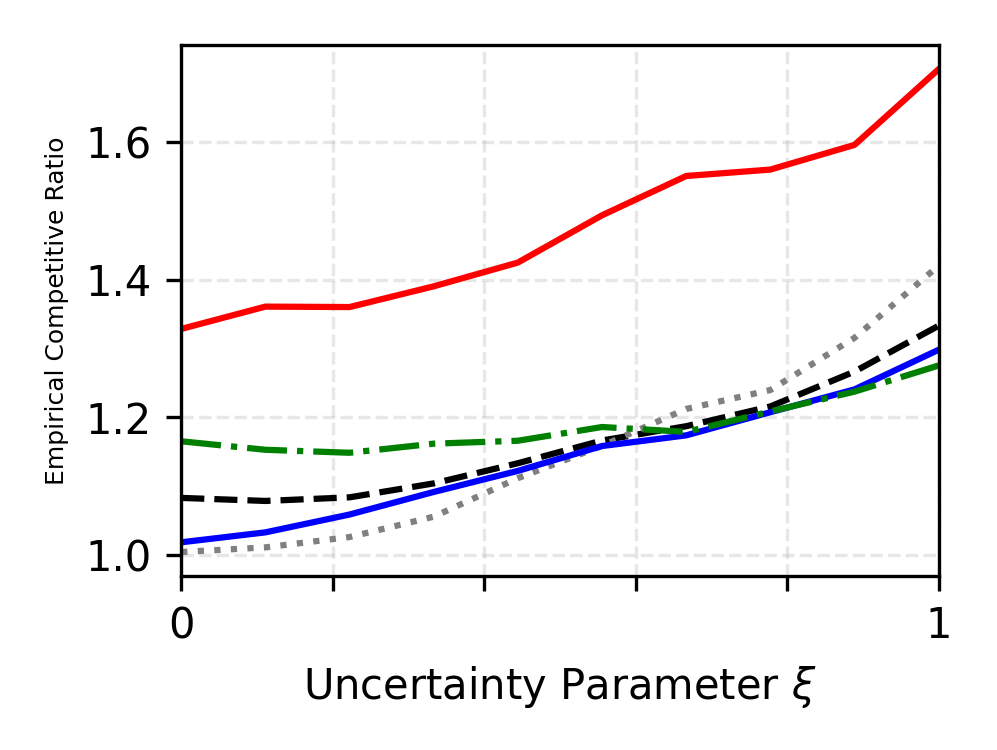}
        \caption{Changing uncertainty}
        \label{fig:uncertainty}
    \end{subfigure}
    \hfill
    \begin{subfigure}{0.24\textwidth}
        \centering
        \includegraphics[width=\textwidth]{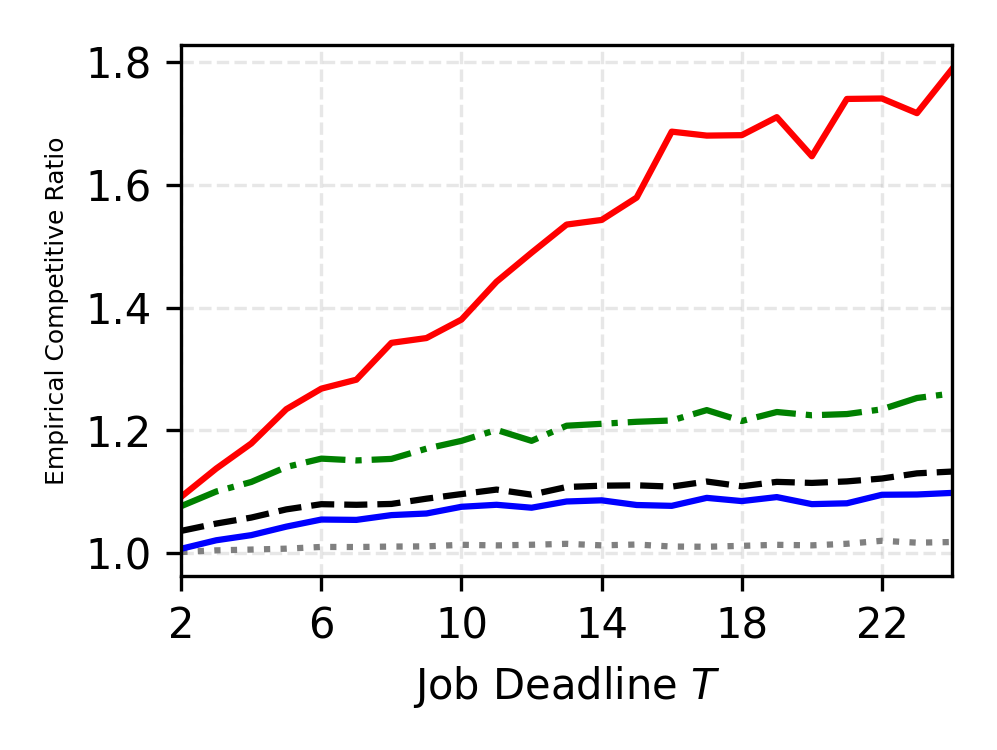}
        \caption{Changing deadline $T$}
        \label{fig:deadline}
    \end{subfigure}
    \hfill
    \begin{subfigure}{0.24\textwidth}
        \centering
        \includegraphics[width=\textwidth]{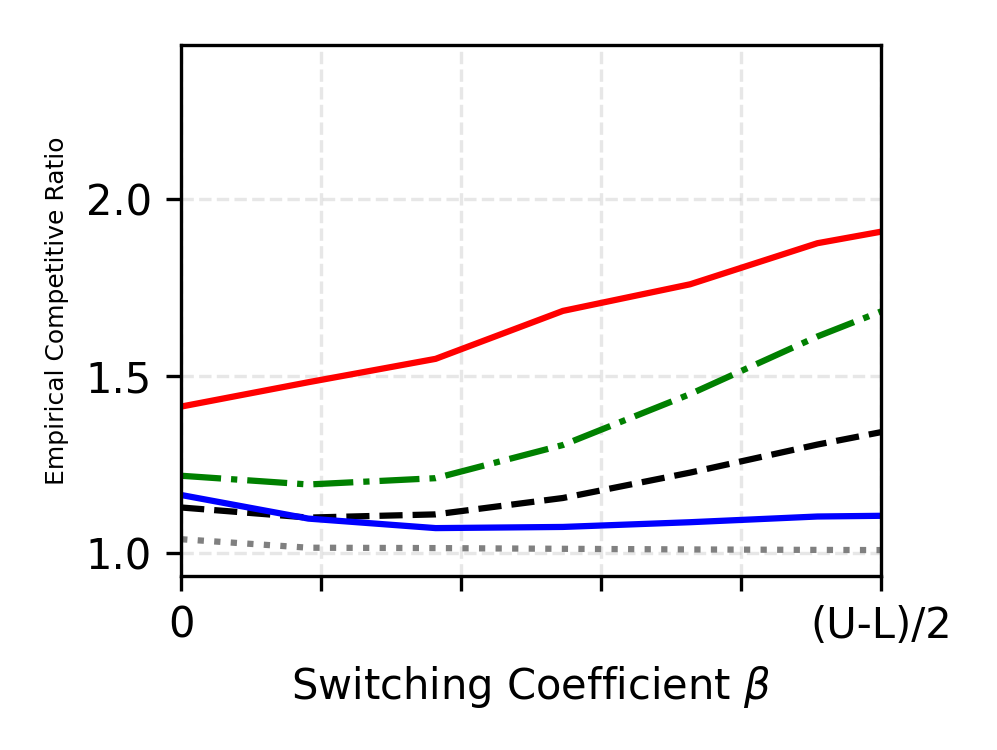}
        \caption{Changing $\beta$}
        \label{fig:beta}
    \end{subfigure}
    \begin{subfigure}{.24\textwidth}
        \centering
        \includegraphics[width=\textwidth]{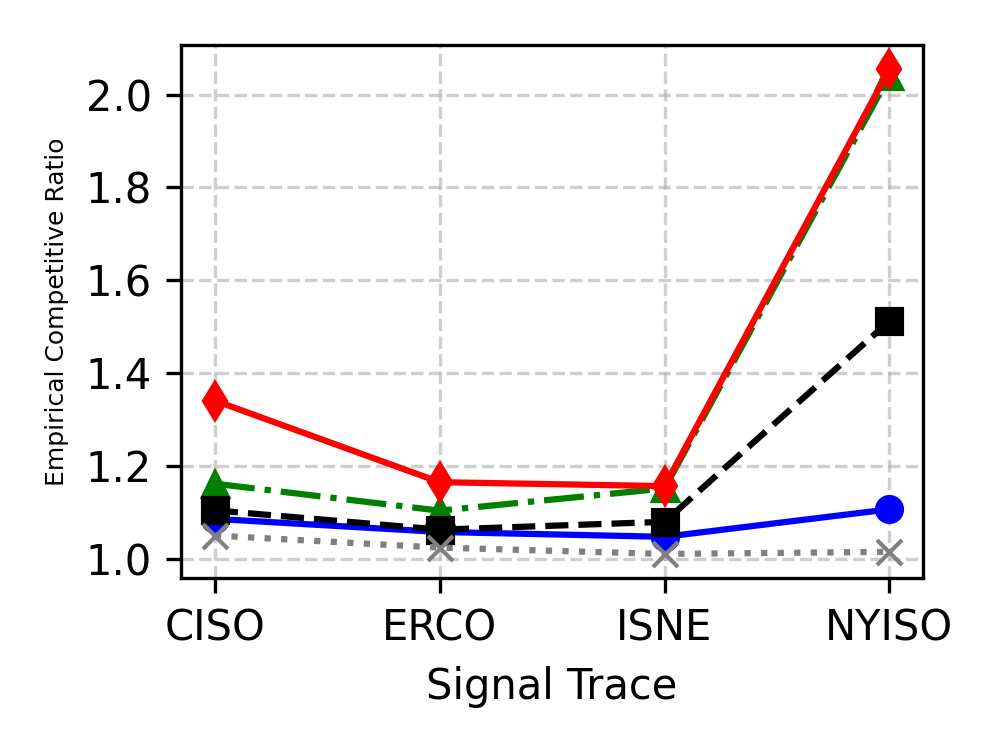}
        \caption{Changing signal trace}
        \label{fig:locations}
    \end{subfigure}
    
    \caption{Experiments with varying parameters.  With the exception of (d), all experiments average over all four signal traces.\\
    (a): Changing the ``width'' of UQ sets and resulting synthetic forecast quality using $\xi$ parameter introduced in \autoref{sec:exp-setup}.\\
    (b): Changing the deadline (i.e., time horizon) $T$ -- as $T$ grows, all algorithms have more flexibility to shift the workload in time.\\ 
    (c): Changing the switching cost parameter $\beta$ (see \autoref{eq:obj}) -- as $\beta$ grows, changes in online decisions are penalized more.\\
    (d): Evaluating the effects of different signal traces (i.e., data sets).  Note that the \texttt{NYISO} trace is used for \textit{electricity cost minimization} experiments, while all other traces are used for \textit{carbon footprint minimization} experiments (see \autoref{sec:exp-setup}).}
    \label{fig:varied}
\end{figure*}

We introduce our experimental results in this section. We find that our \texttt{UQ-Advice} algorithm outperforms baselines by 12.6\% on average and 26.15\% in the 95th percentile when using empirical competitive ratio as a metric of comparison -- see \sref{Table}{tab:results} for a detail of these top-level results across all signal traces.

In \autoref{fig:cdfs}, we present the results of two objective settings we considered: carbon cost in 3 different energy grids as in \autoref{fig:carbon_cdf} and energy price in \autoref{fig:elec_cdf}. We observe that our \texttt{UQ-Advice} algorithm is able to largely outperform the uncertainty-agnostic \texttt{RO-Advice} with $\lambda=0.5$ in both cases and attains performance closer to that of the hindsight-optimized \texttt{RO-Advice} with $\lambda = \lambda^\star$. Note that the hindsight-optimized baseline is a strong one -- namely, the instances considered in \autoref{fig:cdfs} are the same instances on which the value of $\lambda$ has been optimized.
In the carbon footprint minimization case, we note improvements of 1.74\% over the naive \ROAdv with $\lambda = 0.5$, 6.33\% over \RORO, and 10.96\% over \texttt{Threshold} while remaining within 3.42\% of \ROAdv with $\lambda = \lambda^\star$. In the electricity price case study, we note even larger improvements of 24.82\% over the naive \ROAdv with $\lambda = 0.5$, 42.66\% over \RORO, and 44.26\% over \texttt{Threshold} while remaining within 11.14\% of \ROAdv with $\lambda = \lambda^\star$ -- notably, in roughly half of all instances, \UQAdv exactly matches the performance of the hindsight-optimal \ROAdv.

\smallskip
\noindent \textbf{Effect of Parameters. \ }
In \autoref{fig:varied}, we show how the tested algorithms in our case study respond to variations of key parameters of the \ROCU problem -- see \autoref{sec:exp-setup} for specifics of how we manipulate these parameters.

$\triangleright$ \textit{Effect of Uncertainty. \ }
In \autoref{fig:uncertainty}, we vary $\xi$ as described in \autoref{sec:exp-setup} and observe that in high uncertainty environments (i.e., as $\xi$ approaches $1$), our \texttt{UQ-Advice} outperforms both the naive and hindsight-optimized versions of \texttt{RO-Advice} and is on average 5.44\% better than our baselines. 
$T$ and $\beta$ are kept constant as they were in the main experiment setting, and results are averaged over all tested traces.
For each individual baseline algorithm, we note a 2.14\% increase over \texttt{RO-Advice} with $\lambda = 0.5$, a 1.32\% increase over \texttt{RO-Advice} with $\lambda = \lambda^\star$, a 4.03\% increase over \texttt{RORO}, and a 22.82\% increase over \texttt{Threshold}. By improving over each baseline, \texttt{UQ-Advice} demonstrates the advantages of integrating information about uncertainty into the decision-making pipeline. 
Recall that the hindsight-optimized $\ROAdv[\lambda = \lambda^\star]$ is using the best combination factor $\lambda$ as found based on the data sets with real uncertainty forecasts -- this effectively simulates a distribution shift scenario where a previously fine-tuned instance of \ROAdv is deployed in a new, higher uncertainty scenario. 

$\triangleright$ \textit{Effect of Deadline. \ }
In \autoref{fig:deadline}, we observe that \texttt{UQ-Advice} continues to outperform the naïve \texttt{RO-Advice} with $\lambda = 0.5$ and the performance ordering of all algorithms is relatively stable over a tested deadline range of $T \in [2,24]$.  Other parameters are kept constant, and results are averaged over all tested traces.  Note that as $T$ grows, all algorithms but particularly the optimal solution have more flexibility (read: more opportunities) to shift the workload in time.  For this reason, the most drastic performance degradations are for the \texttt{Threshold} and \RORO baselines, which are both conservative and complete the workload at suboptimal signal values, indicated by the widening gap between them and \OPT.  Compared to all baselines, \UQAdv improves by an average of 7.77\%; which breaks down to a 2.48\% improvement over \texttt{RO-Advice} with $\lambda = 0.5$, a 10.06\% improvement over \texttt{RORO}, and a 28.19\% improvement over \texttt{Threshold}. Additionally, \UQAdv lies  within 5.75\% of the hindsight optimized \ROAdv with $\lambda = \lambda ^\star$. 

$\triangleright$ \textit{Effect of Switching Cost. \ }
In \autoref{fig:beta}, we vary the switching cost parameter $\beta$ to observe how this impacts the performance of each algorithm, keeping other parameters constant and averaging over all tested traces.
As $\beta$ grows, the penalty of changing decisions across adjacent time steps becomes more severe.
Interestingly, we find that \UQAdv's performance with respect to baselines improves noticeably as $\beta$ increases towards the (trace-dependent) upper bound of $\nicefrac{U-L}{2}$.  In our default experiment, we set $\beta = 20$, which is towards the lower end of the tested range -- as $\beta$ grows, the gap between \UQAdv and \ROAdv grows noticeably, indicating that information about the uncertainty in forecasts is particularly useful when e.g., it is undesirable for workload decisions to fluctuate significantly.
On average, \UQAdv improves on all baselines by an average of 20.52\%. Individually, \UQAdv improves on \ROAdv with $\lambda = 0.5$ by 17.39\%, \texttt{RORO} by 34.12\%, and \texttt{Threshold} by 41.43\%. Additionally, \UQAdv remained within 9.94\% of the hindsight optimized \ROAdv with $\lambda=\lambda^\star$ across all settings of $\beta$.

$\triangleright$ \textit{Effect of Signal Traces. \ }
Finally, we consider the effect of different signal characteristics by disaggregating performance on each tested data set.  In these experiments, all other parameters are set to their default values (see \autoref{sec:exp-setup}).
Recall that the objective on each trace is not necessarily the same -- the \texttt{NYISO} trace contains electricity price data, so the objective is to minimize the cost of electricity, while the other traces (\texttt{CAISO}, \texttt{ERCOT}, \texttt{ISONE}) all contain carbon intensity data, so the objective is to minimize carbon footprint.

In \autoref{fig:locations}, perhaps the most notable difference is the relative performance of different algorithms on the \texttt{NYISO} trace compared to \textit{all other traces} -- due to the difference in objectives, this corresponds to an observation that many of the baselines perform noticeably worse on the electricity cost minimization task.  In this respect, the result in \autoref{fig:locations} and the observations from \autoref{fig:elec_cdf} are a strength for the \UQAdv algorithm -- using information about the uncertainty in the electricity price forecasts, \UQAdv is able to stay substantially closer to the hindsight optimal \ROAdv with $\lambda = \lambda ^ \star$ while other baselines degrade.
To further underscore this point, prior work has identified that electricity price signals can be more variable and difficult to predict compared to (average) carbon intensity signals~\cite{Bovornkeeratiroj:25}, amplifying the benefits of uncertainty-quantification and integrating UQ forecasts into decision-making.

Across all tested traces, \UQAdv continues to outperform the uncertainty-agnostic \texttt{RO-Advice} with $\lambda = 0.5$, improving by an average of 7.62\%.  \UQAdv also improves on \RORO by an average of 15.47\%, and improves on \texttt{Threshold} by an average of 19.42\%, and remains within 4.89\% of the performance of the hindsight optimal \ROAdv with $\lambda = \lambda ^ \star$.

\smallskip
\noindent \textbf{Key Takeaways. \ }
Compared to \ROAdv with a fixed ``trust parameter'' ($\lambda = 0.5$), our experimental results demonstrate that \UQAdv consistently achieves performance significantly closer to a baseline that has been fine-tuned with the perfect ``trust parameter'' in hindsight (\ROAdv with $\lambda = \lambda^\star$).  For practitioners, this provides a valuable and concrete benefit of proactively considering uncertainty-quantified forecasts in decision-making, as most deployed systems are not likely to undergo regular fine-tuning.
\UQAdv also outperforms standard robust baselines (i.e., designed using classic competitive analysis) for the problem, and the relative ordering of algorithm performances remains stable across diverse real-world signal traces.

In synthetically-generated scenarios with high forecast uncertainty (see \autoref{fig:uncertainty}),
\UQAdv's ability to leverage UQ information allows it to even surpass the approach that was perfectly-tuned (in hindsight) on UQ forecasts provided in the traces.  These scenarios can be thought of as ``distribution shift'' simulations that approximate what would happen if an algorithm fine-tuned in one scenario is suddenly faced with a substantially different (i.e., out-of-distribution) environment.  In these cases, \UQAdv's consideration of uncertainty-quantified forecasts allows it to automatically adapt and match the best robust algorithm that does not use forecasts (\RORO).

\section{Conclusion}
\label{sec:conclusion}

We introduce \UQAdv, a novel learning-augmented algorithm for online workload shifting that systematically leverages modern uncertainty-quantified (UQ) predictors.  Ours is the first work, to our knowledge, to systematically consider complex multi-time-step uncertainty-quantification for online algorithm design.
At its core is the decision uncertainty score (\DUS) measure, which translates forecast uncertainty into decision variability, allowing the algorithm to automatically balance trust between learned advice and a robust online baseline without manual tuning. We provide rigorous theoretical guarantees for the algorithm's 
consistency, robustness, and a new UQ-robustness metric, showing that performance degrades gracefully as uncertainty increases. Trace-driven experiments on real-world carbon and electricity price data confirm that \UQAdv outperforms existing methods, particularly in high-uncertainty scenarios, demonstrating the practical benefits of integrating uncertainty-aware predictions into online algorithms.

Several interesting avenues remain for future work.  We introduce the decision uncertainty score (\DUS) measure as a useful tool to translate the complex uncertainty sets defined by multi-time-step UQ forecasts into an actionable metric for decision-making and theoretical analysis -- however, the optimization problem that defines it is non-convex and thus difficult to analyze in a theoretically rigorous manner.  For this and other online algorithms problems, it would be interesting to explore other ideas that similarly ``map'' from complex uncertainty-quantified advice to some simpler, actionable metrics for decision making, without sacrificing properties that facilitate analysis.
There also remain interesting applications in sustainability that are not fully modeled by the \ROCU problem considered in this paper, such as signal-aware workload shifting for multiple simultaneous workloads with heterogeneous deadlines (particularly relevant for data centers).  We speculate that some of the high-level ideas considered in this paper may generalize to such settings, and it would be interesting to consider how uncertainty-quantified forecasts might inform decision-making in these more complicated settings.

\begin{acks}
This research is supported by National Science Foundation grants 2045641, 2325956, 2512128, and 2533814, and the U.S. Department of Energy, Office of Science, Office of Advanced Scientific Computing Research, Department of Energy Computational Science Graduate Fellowship under Award Number DE-SC0024386.

We thank Nico Christianson for several helpful discussions. We also thank Saud Alghumayjan, Ming Yi, and Bolun Xu~\cite{Alghumayjan:24}, and Amy Li, Sihang Liu, and Yi Ding~\cite{Li:24} for providing data that we used in our evaluation. 
\end{acks}

\bibliographystyle{ACM-Reference-Format}
\bibliography{main}

\appendix
\onecolumn
\section*{Appendix}

\section{Deferred Proofs from \autoref{sec:leverage}} \label{apx:proof-leverage}

In what follows, we provide the deferred proof of \autoref{thm:global-optimum-dus}, which states that if $y^\star$ denotes the true global optimal value of $\DUS(\mathcal{U}, \mathbf{p})$ (in \eqref{eq:dus}), a global non-convex optimizer can find a solution $\hat{y}$ that satisfies the following inequality:
\[
y^\star \leq \hat{y} + \varepsilon,
\]
for any fixed $\varepsilon > 0$ in finite time, using $O\left( \left( \frac{\sqrt{T} \cdot \text{diam}(\mathcal{U})}{2 \lambda \cdot \varepsilon} \right)^T \right)$ iterations~\cite{Malherbe:17,  Bachoe:21}.

\begin{proof}[Proof of \autoref{thm:global-optimum-dus}]
Let $y^\star$ be defined as follows:
\[
y^\star \coloneqq \DUS(\mathcal{U}, \mathbf{p}) \coloneqq \max_{\mathbf{z} \in \mathcal{U}} \Vert \OPT(\hat{\mathbf{p}}) - \OPT(\mathbf{z})\Vert_1,
\]
for arbitrary $\mathcal{U}$ and $\mathbf{p}$.

We start by showing that the function $f(\mathbf{z}) = \Vert \OPT(\hat{\mathbf{p}}) - \OPT(\mathbf{z})\Vert_1$ is Lipschitz continuous.  
To do so, consider the mapping $\OPT(\mathbf{z}) : \mathcal{U} \to [0,1]^T$.  First, by the first-order optimality condition (and since \eqref{eq:obj} is strongly convex, admitting a unique optimal solution), we have the following useful facts.

Consider two arbitrary inputs $\mathbf{z}_1$ and $\mathbf{z}_2$.  Let $\mathbf{x}^\star_1 = \OPT(\mathbf{z}_1)$ and $\mathbf{x}^\star_2 = \OPT(\mathbf{z}_2)$ be the corresponding optimal solutions to the \ROCU problem.  Let $g(\mathbf{x}, \mathbf{z})$ denote the objective in \eqref{eq:obj}.  Note that the subdifferential of $g$ with respect to $\mathbf{x}$ is $\partial_{\mathbf{x}} g(\mathbf{x}, \mathbf{z}) = \mathbf{z} + \partial R(\mathbf{x})$, where $R(\mathbf{x}) = \sum_{t=1}^{T+1} \beta \vert x_t - x_{t-1} \vert + \lambda \Vert \mathbf{x} \Vert_2^2$ collects the non-linear terms in the objective.

By the first-order optimality condition, at the point of optimality there must exist a subgradient $\mathbf{j} \in \partial_{\mathbf{x}} g(\mathbf{x}^\star_1, \mathbf{z}_1)$ such that
\[\langle \mathbf{j}, \mathbf{y} - \mathbf{x}^\star_1 \rangle \geq 0, \quad \forall \mathbf{y} \in \mathcal{X},\]
where $\mathcal{X}$ is the feasible set of decisions.  

Then for $\mathbf{x}^\star_1$ and $\mathbf{x}^\star_2$, letting $\mathbf{h}_1 \in \partial_{\mathbf{x}} R(\mathbf{x}^\star_1)$ and $\mathbf{h}_2 \in \partial_{\mathbf{x}} R(\mathbf{x}^\star_2)$, we have
\begin{align*}
    \langle \mathbf{z}_1 + \mathbf{h}_1, \mathbf{x}^\star_2 - \mathbf{x}^\star_1 \rangle & \geq 0, \\
    \langle \mathbf{z}_2 + \mathbf{h}_2, \mathbf{x}^\star_1 - \mathbf{x}^\star_2 \rangle & \geq 0.
\end{align*}
Adding these two inequalities, we have
\[\langle \mathbf{z}_1 - \mathbf{z}_2, \mathbf{x}^\star_2 - \mathbf{x}^\star_1 \rangle + \langle \mathbf{h}_1 - \mathbf{h}_2, \mathbf{x}^\star_2 - \mathbf{x}^\star_1 \rangle \geq 0.\]
Rearranging, we have
\[\langle \mathbf{z}_1 - \mathbf{z}_2, \mathbf{x}^\star_2 - \mathbf{x}^\star_1 \rangle \geq - \langle \mathbf{h}_1 - \mathbf{h}_2, \mathbf{x}^\star_2 - \mathbf{x}^\star_1 \rangle.\]

Recall that the definition of strong convexity states that a function $F(\mathbf{x})$ is $\mu$-strongly convex if for all $\mathbf{x}, \mathbf{y}$ and for any corresponding subgradients $\mathbf{h}_x \in \partial F(\mathbf{x})$ and $\mathbf{h}_y \in \partial F(\mathbf{y})$, we have
\[\langle \mathbf{h}_x - \mathbf{h}_y, \mathbf{x} - \mathbf{y} \rangle \geq \mu \Vert \mathbf{x} - \mathbf{y} \Vert_2^2.\]
Since $R(\mathbf{x})$ is strongly convex because it is the sum of a convex function and a $2\lambda$-strongly convex function, we can apply this definition to obtain the following:
\[\langle \mathbf{h}_1 - \mathbf{h}_2, \mathbf{x}^\star_2 - \mathbf{x}^\star_1 \rangle \geq 2\lambda \Vert \mathbf{x}^\star_2 - \mathbf{x}^\star_1 \Vert_2^2.\]
Plugging this into the previous inequality, we have
\[\langle \mathbf{z}_1 - \mathbf{z}_2, \mathbf{x}^\star_2 - \mathbf{x}^\star_1 \rangle \geq 2\lambda \Vert \mathbf{x}^\star_2 - \mathbf{x}^\star_1 \Vert_2^2.\]
Applying the Cauchy-Schwarz inequality to the left-hand side, we have
\[\Vert \mathbf{z}_1 - \mathbf{z}_2 \Vert_2 \cdot \Vert \mathbf{x}^\star_2 - \mathbf{x}^\star_1 \Vert_2 \geq 2\lambda \Vert \mathbf{x}^\star_2 - \mathbf{x}^\star_1 \Vert_2^2.\]
If $\mathbf{x}^\star_2 = \mathbf{x}^\star_1$, the condition is trivially satisfied.  Otherwise, we can divide both sides by $\Vert \mathbf{x}^\star_2 - \mathbf{x}^\star_1 \Vert_2$ and rearrange to obtain
\[\Vert \mathbf{x}^\star_2 - \mathbf{x}^\star_1 \Vert_2 \leq \frac{1}{2\lambda} \Vert \mathbf{z}_1 - \mathbf{z}_2 \Vert_2.\]
This shows that the mapping $\OPT(\mathbf{z})$ is $\frac{1}{2\lambda}$-Lipschitz continuous.

Then, we show that $f(\mathbf{z}) = \Vert \OPT(\hat{\mathbf{p}}) - \OPT(\mathbf{z})\Vert_1$ is Lipschitz continuous.  We want to show that there exists a constant $K$ such that for all $\mathbf{z}_1, \mathbf{z}_2 \in \mathcal{U}$,
\[\vert f(\mathbf{z}_1) - f(\mathbf{z}_2) \vert \leq K \Vert \mathbf{z}_1 - \mathbf{z}_2 \Vert_2.\]
By the reverse triangle inequality, we have:
\begin{align*}
    \vert f(\mathbf{z}_1) - f(\mathbf{z}_2) \vert & = \vert \Vert \OPT(\hat{\mathbf{p}}) - \OPT(\mathbf{z}_1)\Vert_1 - \Vert \OPT(\hat{\mathbf{p}}) - \OPT(\mathbf{z}_2)\Vert_1 \vert \\
    & \leq \Vert (\OPT(\hat{\mathbf{p}}) - \OPT(\mathbf{z}_1)) - (\OPT(\hat{\mathbf{p}}) - \OPT(\mathbf{z}_2)) \Vert_1 \\
    & = \Vert \OPT(\mathbf{z}_2) - \OPT(\mathbf{z}_1) \Vert_1.
\end{align*}
Using the result that $\OPT(\mathbf{z})$ is $\frac{1}{2\lambda}$-Lipschitz continuous and equivalence of norms, we have
\[\Vert \OPT(\mathbf{z}_2) - \OPT(\mathbf{z}_1) \Vert_1 \leq \sqrt{T} \Vert \OPT(\mathbf{z}_2) - \OPT(\mathbf{z}_1) \Vert_2 \leq \frac{\sqrt{T}}{2\lambda} \Vert \mathbf{z}_2 - \mathbf{z}_1 \Vert_2.\]
Thus, $f(\mathbf{z})$ is Lipschitz continuous with Lipschitz constant $K = \frac{\sqrt{T}}{2\lambda}$.

Having shown that $f(\mathbf{z})$ is a non-concave Lipschitz continuous function, we can apply standard results on global non-convex optimization~\cite{Malherbe:17,  Bachoe:21} to conclude that a global non-convex optimizer is able to find a solution $\hat{y}$ that satisfies the following inequality:
\[y^\star \leq \hat{y} + \varepsilon,\]
for any fixed $\varepsilon > 0$ in finite time, using $O\left( \left( \frac{K \cdot \text{diam}(\mathcal{X})}{\varepsilon} \right)^T \right)$ iterations, where $\text{diam}(\mathcal{X})$ is the diameter of the feasible decision set $\mathcal{X}$, and $K$ is the Lipschitz constant of the non-concave objective.  Plugging in the Lipschitz constant $K = \frac{\sqrt{T}}{2\lambda}$ and letting $\mathcal{U}$ denote the feasible set of $\mathbf{z}$ gives the result.

\end{proof}

\section{Deferred Proofs from \autoref{sec:alg}} \label{apx:proof-alg}

In what follows, we provide deferred proofs for the theorems in \autoref{sec:alg}.

\subsection{Proof of \autoref{thm:roro-comp-rocu}} \label{apx:proof-roro}
We begin by proving \autoref{thm:roro-comp-rocu}, which states that the \RORO algorithm given by \cite{Lechowicz:24} and summarized in \sref{Algorithm}{alg:roro} is $\alpha$-competitive for \ROCU, where $\alpha$ is defined as:
\begin{equation}
    \alpha \coloneqq \frac{T(\alpha_{\RORO} \cdot p_{\min} + \lambda)}{T p_{\min} + \lambda},
\end{equation}
where $\alpha_{\RORO}$ is defined in \eqref{eq:alpha_roro}, $T$ is the time horizon, $p_{\min}$ is the minimum price, and $\lambda$ is the regularization parameter.  Note that as $\lambda \rightarrow 0$, we have $\alpha \rightarrow \alpha_{\RORO}$.

\begin{proof}[Proof of \autoref{thm:roro-comp-rocu}]
We prove the result by extending results in \cite{Lechowicz:24} to account for the regularization term.  Let $\RORO(\mathcal{I})$ be the cost incurred by \RORO on instance $\mathcal{I}$, and let $\OPT(\mathcal{I})$ be the optimal offline cost for the same instance.  

First, we have the following lemma to lower bound the cost of \OPT:
\begin{lemma} \label{lem:opt-lb}
    For any instance $\mathcal{I}$ of \ROCU, we have:
    \begin{equation}
        \Cost(\OPT(\mathcal{I})) \geq \phi(w_j) - \beta + \frac{\lambda}{T}
    \end{equation}
    where $w_j$ is the final utilization of \RORO (before the compulsory trade).
\end{lemma}
\begin{proof}
First, by \cite[Lemma 1]{Lechowicz:24}, we know that the optimal offline solution's cost due to the first two terms of the objective function (i.e., excluding the regularization term) is at least $\textbf{OPT}(\mathcal{I}) \phi(w_j) - \beta$.  

Next, we need to account for the regularization term $\lambda \Vert \mathbf{x} \Vert_2^2$.  To do so, note that by the constraints of the problem, we know that $\sum_{t=1}^T x_t = 1$. 
Then we have:
\begin{align*}
\Vert \mathbf{x} \Vert_2^2 &= \sum_{t=1}^T x_t^2,\\
&\geq \sum_{t=1}^T \left(\frac{1}{T}\right)^2 = \frac{1}{T},
\end{align*}
where the second inequality follows because the sum of squares is minimized when all $x_t$ are equal.
Thus, the regularization term contributes at least $\lambda / T$ to the cost of \OPT.  Combining this with the previous lower bound gives the result.
\end{proof}

Next, we have the following lemma to upper bound the cost of \RORO:
\begin{lemma} \label{lem:roro-ub}
    For any instance $\mathcal{I}$ of \ROCU, we have:
    \begin{equation}
        \Cost(\RORO(\mathcal{I})) \leq \int_0^{w_j} \phi(u) du + \beta + (1-w_j) p_{\max} + \lambda
    \end{equation}
    where $w_j$ is the final utilization of \RORO (before the compulsory trade).
\end{lemma}
\begin{proof}
First, by \cite[Lemma 2]{Lechowicz:24}, we know that the cost of \RORO due to the first two terms of the objective function (i.e., excluding the regularization term) is at most $\int_0^{w_j} \phi(u) du + \beta + (1-w_j) p_{\max}$.  Next, we need to account for the regularization term $\lambda \Vert \mathbf{x} \Vert_2^2$.  To do so, note that by the constraints of the problem, we know that $\sum_{t=1}^T x_t = 1$.
Then we have:
\begin{align*}
\Vert \mathbf{x} \Vert_2^2 &= \sum_{t=1}^T x_t^2,\\
&\leq 1^2 = 1,
\end{align*}
where the second inequality follows because the sum of squares is maximized when all $x_t$ are zero except one $x_t$ which is equal to 1.
Thus, the regularization term contributes at most $\lambda$ to the cost of \RORO.  Combining this with the previous upper bound gives the result.
\end{proof}

By the proof of \cite[Theorem 1]{Lechowicz:24}, we know that:
\begin{equation}
    \int_0^{w} \phi(u) du + \beta + (1-w) p_{\max} = \alpha_{\RORO} (\phi(w) - \beta), \quad \forall w \in [0,1].
\end{equation}
Combining this with \sref{Lemma}{lem:roro-ub}, we have:
\begin{align*}
    \Cost(\RORO(\mathcal{I})) &\leq \alpha_{\RORO} (\phi(w_j) - \beta) + \lambda.\\
\end{align*}
Finally, combining this with \sref{Lemma}{lem:opt-lb}, we have:
\begin{align*}
    \frac{\Cost(\RORO(\mathcal{I}))}{\Cost(\OPT(\mathcal{I}))} &\leq \frac{\alpha_{\RORO} (\phi(w_j) - \beta) + \lambda}{\phi(w_j) - \beta + \frac{\lambda}{T}}, \\
    &\leq \frac{\alpha_{\RORO} p_{\min} + \lambda}{p_{\min} + \frac{\lambda}{T}},\\
    &\leq \frac{T(\alpha_{\RORO} p_{\min} + \lambda)}{T p_{\min} + \lambda} = \alpha.
\end{align*}
This completes the proof.
\end{proof}

\subsection{Proof of \autoref{thm:uq-consistency}} \label{apx:proof-consistency}

Next, we prove \autoref{thm:uq-consistency}, which states that \UQAdv is $\eta$-consistent under the assumption that the point predictions are perfect (i.e., $\hat{\mathbf{p}} = \mathbf{p}$), where:
\begin{equation}
    \eta \coloneqq 1 + (\alpha - 1) \frac{\DUS(\mathcal{U}, \mathbf{\hat{p}})}{2} ,
\end{equation}
where $\alpha$ is the competitive ratio of \RORO for \ROCU given in \autoref{thm:roro-comp-rocu}. 

\begin{proof} [Proof of \autoref{thm:uq-consistency}]
To show consistency, first recall that the point predictions are perfect, meaning that $\hat{p}_t = p_t$ for all $t$.  
In this case, note that $\OPT(\mathcal{I}) = \{\hat{x}_t\}_{t=1}^T$, since the optimal decisions under perfect predictions (i.e., $\OPT(\hat{\mathbf{p}})$) are exactly correct.

Then, from the definition of \UQAdv, we have the following upper bound on the total cost incurred by \UQAdv.  We let $r \coloneqq \Vert \mathbf{\hat{x}} \Vert_2^2$ denote the regularization cost paid by \OPT for brevity.
\begin{align*}
    \Cost(\UQAdv(\mathcal{I})) &= \sum_{t=1}^{T} p_t x_t + \sum_{t=1}^{T+1} \beta \vert x_t - x_{t-1} \vert + \lambda \Vert \mathbf{x} \Vert_2^2,\\
    &= \sum_{t=1}^{T} p_t(\gamma \hat x_t + (1-\gamma) \tilde x_t) + \sum_{t=1}^{T+1} \beta|\gamma \hat x_t + (1-\gamma) \tilde x_t - \gamma \hat x_{t-1} - (1-\gamma) \tilde x_{t-1}| + \lambda \Vert \gamma \hat{\mathbf{x}} + (1-\gamma) \tilde{\mathbf{x}} \Vert_2^2, \\
    &\leq \gamma \sum_{t=1}^{T} p_t \hat x_t + (1-\gamma) \sum_{t=1}^T p_t \tilde x_t + \sum_{t=1}^{T+1} \beta|\gamma (\hat x_t - \hat x_{t-1}) + (1-\gamma) (\tilde x_t - \tilde x_{t-1})| + \lambda (\gamma^2 r + (1-\gamma)^2 \Vert \tilde{\mathbf{x}} \Vert_2^2 + 2\gamma(1-\gamma) \langle \hat{\mathbf{x}}, \tilde{\mathbf{x}} \rangle) \\
\end{align*}
By triangle inequality, the switching cost term is upper bounded as follows:
\begin{align*}
    \Cost(\UQAdv(\mathcal{I})) &\leq \gamma \sum_{t=1}^{T} p_t \hat x_t + (1-\gamma) \sum_{t=1}^T p_t \tilde x_t + \gamma \sum_{t=1}^{T+1} \beta|\hat x_t + \hat x_{t-1}| + (1-\gamma)\sum \beta | \tilde x_t - \tilde x_{t-1}| + \lambda (\gamma^2 r + (1-\gamma)^2 \Vert \tilde{\mathbf{x}} \Vert_2^2 + 2\gamma(1-\gamma) \langle \hat{\mathbf{x}}, \tilde{\mathbf{x}} \rangle) \\
\end{align*}
Using Cauchy-Schwarz on the final term, the above simplifies to:
\begin{align*}
    \Cost(\UQAdv(\mathcal{I})) &\leq \gamma \sum_{t=1}^{T} p_t \hat x_t + (1-\gamma) \sum_{t=1}^T p_t \tilde x_t + \gamma \sum_{t=1}^{T+1} \beta|\hat x_t + \hat x_{t-1}| + (1-\gamma)\sum \beta | \tilde x_t - \tilde x_{t-1}| + \lambda (\gamma \sqrt{r} + 1 - \gamma)^2 \\
\end{align*}
With a slight abuse of notation, let $\OPT_{\text{bought}}(\mathcal{I})$, $\OPT_{\text{switch}}(\mathcal{I})$, and $\OPT_{\text{reg}}(\mathcal{I}) \coloneqq r \cdot \lambda$ denote the cost of \OPT due to the buying, switching, and regularization costs, respectively.  Similarly, let $\RORO_{\text{bought}}(\mathcal{I})$, $\RORO_{\text{switch}}(\mathcal{I})$, and $\RORO_{\text{reg}}(\mathcal{I}) \coloneqq 1 \cdot \lambda$ denote the (worst-case) cost of \RORO due to the buying, switching, and regularization costs, respectively.  Then, we have:
\begin{align*}
    \Cost(\UQAdv(\mathcal{I})) &\leq \gamma \OPT_{\text{bought}}(\mathcal{I}) + (1-\gamma) \RORO_{\text{bought}}(\mathcal{I}) + \gamma \OPT_{\text{switch}}(\mathcal{I}) + (1-\gamma)\RORO_{\text{switch}}(\mathcal{I}) + \lambda \left( \gamma \sqrt{r} + \left( 1 - \gamma \right) \sqrt{1} \right)^2 \\
\end{align*}
By Jensen's inequality, we have:
$$\left( \gamma \sqrt{r} + (1-\gamma) \sqrt{1} \right)^2 \leq \gamma r + (1-\gamma) 1.$$
Combining this with the above gives:
\begin{align*}
    \Cost(\UQAdv(\mathcal{I})) &\leq \gamma \OPT_{\text{bought}}(\mathcal{I}) + (1-\gamma) \RORO_{\text{bought}}(\mathcal{I}) + \gamma \OPT_{\text{switch}}(\mathcal{I}) + (1-\gamma)\RORO_{\text{switch}}(\mathcal{I}) + \lambda (\gamma r + (1-\gamma) 1) \\
    &\leq \gamma (\OPT_{\text{bought}}(\mathcal{I}) + \OPT_{\text{switch}}(\mathcal{I}) + \OPT_{\text{reg}}(\mathcal{I})) + (1-\gamma) (\RORO_{\text{bought}}(\mathcal{I}) + \RORO_{\text{switch}}(\mathcal{I}) + \RORO_{\text{reg}}(\mathcal{I})) \\
    &\leq \gamma \cdot \Cost(\OPT(\mathcal{I})) + (1-\gamma) \cdot \Cost(\RORO(\mathcal{I})).
\end{align*}
Since \RORO is $\alpha$-competitive for \ROCU by \autoref{thm:roro-comp-rocu}, we have:
\begin{align*}
    \Cost(\UQAdv(\mathcal{I})) &\leq \gamma \cdot \Cost(\OPT(\mathcal{I})) + (1-\gamma) \cdot \alpha \cdot \Cost(\OPT(\mathcal{I})), \\
    &= (\gamma + (1-\gamma) \alpha) \cdot \Cost(\OPT(\mathcal{I})).
\end{align*}
In terms of \DUS, this is equivalently stated as:
\begin{align}
    \Cost(\UQAdv(\mathcal{I})) &\leq \left( 1 + (\alpha -1) \frac{\DUS(\mathcal{U}, \mathbf{\hat{p}})}{2} \right) \cdot \Cost(\OPT(\mathcal{I})).
\end{align}
This completes the proof.
\end{proof}

\subsection{Proof of \autoref{thm:uq-robustness}} \label{apx:proof-robustness}

Next, we prove \autoref{thm:uq-robustness}, which states that under the assumption that \textit{both} the point predictions and the uncertainty quantification are completely (i.e., adversarially) incorrect, \UQAdv is $\zeta$-robust, where:
    \begin{equation}
        \zeta \coloneqq \left( 1 - \frac{\DUS(\mathcal{U}, \hat{\mathbf{p}})}{2} \right) \cdot \frac{p_{\max} + 2\beta + \lambda}{p_{\min} + \frac{\lambda}{T}} + \left( \frac{\DUS(\mathcal{U}, \hat{\mathbf{p}})}{2} \right) \cdot \alpha,
    \end{equation}
    where $\alpha$ is the competitive ratio of \RORO for \ROCU given in \autoref{thm:roro-comp-rocu}. 

\begin{proof} [Proof of \autoref{thm:uq-robustness}]
To show robustness, first recall that \textit{both} the point predictions and the uncertainty estimates are completely (i.e., adversarially) incorrect. 

In this case, we denote $\ADV(\mathcal{I}) = \{\hat{x}_t\}_{t=1}^T$, as the decisions that assumed the point predictions (i.e., $\OPT(\hat{\mathbf{p}})$) were exactly correct.  Note that the worst-case cost of any $\ADV(\mathcal{I})$ is upper bounded as follows:
\begin{align*}
    \Cost(\ADV(\mathcal{I})) &\leq p_{\max} + 2\beta + \lambda,
\end{align*}
where the first term is the maximum possible buying cost (i.e., buying everything at the highest price), the second term is the maximum possible switching cost, and the third term is the maximum possible regularization cost.

Then, from the definition of \UQAdv, we have the following upper bound on the total cost incurred by \UQAdv. 
\begin{align*}
    \Cost(\UQAdv(\mathcal{I})) &= \sum_{t=1}^{T} p_t x_t + \sum_{t=1}^{T+1} \beta \vert x_t - x_{t-1} \vert + \lambda \Vert \mathbf{x} \Vert_2^2,\\
    &= \sum_{t=1}^{T} p_t(\gamma \hat x_t + (1-\gamma) \tilde x_t) + \sum_{t=1}^{T+1} \beta|\gamma \hat x_t + (1-\gamma) \tilde x_t - \gamma \hat x_{t-1} - (1-\gamma) \tilde x_{t-1}| + \lambda \Vert \gamma \hat{\mathbf{x}} + (1-\gamma) \tilde{\mathbf{x}} \Vert_2^2, \\
    &\leq \gamma \sum_{t=1}^{T} p_t \hat x_t + (1-\gamma) \sum_{t=1}^T p_t \tilde x_t + \sum_{t=1}^{T+1} \beta|\gamma (\hat x_t - \hat x_{t-1}) + (1-\gamma) (\tilde x_t - \tilde x_{t-1})| + \lambda (\gamma^2 + (1-\gamma)^2 \Vert \tilde{\mathbf{x}} \Vert_2^2 + 2\gamma(1-\gamma) \langle \hat{\mathbf{x}}, \tilde{\mathbf{x}} \rangle) \\
\end{align*}
By triangle inequality, the switching cost term is upper bounded as follows:
\begin{align*}
    \Cost(\UQAdv(\mathcal{I})) &\leq \gamma \sum_{t=1}^{T} p_t \hat x_t + (1-\gamma) \sum_{t=1}^T p_t \tilde x_t + \gamma \sum_{t=1}^{T+1} \beta|\hat x_t + \hat x_{t-1}| + (1-\gamma)\sum \beta | \tilde x_t - \tilde x_{t-1}| + \lambda (\gamma^2 + (1-\gamma)^2 \Vert \tilde{\mathbf{x}} \Vert_2^2 + 2\gamma(1-\gamma) \langle \hat{\mathbf{x}}, \tilde{\mathbf{x}} \rangle) \\
\end{align*}
Using Cauchy-Schwarz on the final term, the above simplifies to:
\begin{align*}
    \Cost(\UQAdv(\mathcal{I})) &\leq \gamma \sum_{t=1}^{T} p_t \hat x_t + (1-\gamma) \sum_{t=1}^T p_t \tilde x_t + \gamma \sum_{t=1}^{T+1} \beta|\hat x_t + \hat x_{t-1}| + (1-\gamma)\sum \beta | \tilde x_t - \tilde x_{t-1}| + \lambda (\gamma + 1 - \gamma)^2 \\
\end{align*}
With a slight abuse of notation, let $\ADV_{\text{bought}}(\mathcal{I})$, $\ADV_{\text{switch}}(\mathcal{I})$, and $\ADV_{\text{reg}}(\mathcal{I}) \coloneqq r \cdot \lambda$ denote the cost of \ADV due to the buying, switching, and regularization costs, respectively.  Similarly, let $\RORO_{\text{bought}}(\mathcal{I})$, $\RORO_{\text{switch}}(\mathcal{I})$, and $\RORO_{\text{reg}}(\mathcal{I}) \coloneqq 1 \cdot \lambda$ denote the (worst-case) cost of \RORO due to the buying, switching, and regularization costs, respectively.  Then, we have:
\begin{align*}
    \Cost(\UQAdv(\mathcal{I})) &\leq \gamma \ADV_{\text{bought}}(\mathcal{I}) + (1-\gamma) \RORO_{\text{bought}}(\mathcal{I}) + \gamma \ADV_{\text{switch}}(\mathcal{I}) + (1-\gamma)\RORO_{\text{switch}}(\mathcal{I}) + \lambda (\gamma + (1-\gamma) 1) \\
    &\leq \gamma (\ADV_{\text{bought}}(\mathcal{I}) + \ADV_{\text{switch}}(\mathcal{I}) + \ADV_{\text{reg}}(\mathcal{I})) + (1-\gamma) (\RORO_{\text{bought}}(\mathcal{I}) + \RORO_{\text{switch}}(\mathcal{I}) + \RORO_{\text{reg}}(\mathcal{I})) \\
    &\leq \gamma \cdot \Cost(\ADV(\mathcal{I})) + (1-\gamma) \cdot \Cost(\RORO(\mathcal{I})).
\end{align*}
Since \RORO is $\alpha$-competitive for \ROCU by \autoref{thm:roro-comp-rocu}, we have:
\begin{align*}
    \Cost(\UQAdv(\mathcal{I})) &\leq \gamma \cdot \Cost(\OPT(\mathcal{I})) + (1-\gamma) \cdot \alpha \cdot \Cost(\OPT(\mathcal{I})).
\end{align*}
Then, since the worst-case for the \ADV term occurs when $\Cost(\ADV(\mathcal{I})) = p_{\max} + 2\beta + \lambda$ and $\Cost(\OPT(\mathcal{I})) \geq p_{\min} + \frac{\lambda}{T}$, we have:
\begin{align*}
    \Cost(\UQAdv(\mathcal{I})) &\leq \gamma \cdot \frac{p_{\max} + 2\beta + \lambda}{p_{\min} + \frac{\lambda}{T}} \cdot \Cost(\OPT(\mathcal{I})) + (1-\gamma) \cdot \alpha \cdot \Cost(\OPT(\mathcal{I})), \\
    &= \left( \gamma \cdot \frac{p_{\max} + 2\beta + \lambda}{p_{\min} + \frac{\lambda}{T}} + (1-\gamma) \cdot \alpha \right) \cdot \Cost(\OPT(\mathcal{I})).
\end{align*}
In terms of \DUS, this is equivalently stated as:
\begin{align}
    \Cost(\UQAdv(\mathcal{I})) &\leq \left( \left( 1 - \frac{\DUS(\mathcal{U}, \hat{\mathbf{p}})}{2} \right) \cdot \frac{p_{\max} + 2\beta + \lambda}{p_{\min} + \frac{\lambda}{T}} + \left( \frac{\DUS(\mathcal{U}, \hat{\mathbf{p}})}{2} \right) \cdot \alpha \right) \cdot \Cost(\OPT(\mathcal{I})).
\end{align}
This completes the proof.
\end{proof}

\subsection{Proof of \autoref{thm:uq-uqrobustness}} \label{apx:proof-uqrobustness}

Finally, we prove \autoref{thm:uq-uqrobustness}, which states that under the assumption that the uncertainty quantification is accurate (i.e., the true signals lie within the uncertainty intervals), but the point predictions may be imperfect, \UQAdv is $\theta$-UQ-robust, where:

{\small
\begin{equation*}
    \theta \coloneqq 1 + \left( \frac{\DUS(\mathcal{U}, \hat{\mathbf{p}})}{2} \right) \left( \alpha - 1 + \left( 1 - \frac{\DUS(\mathcal{U}, \hat{\mathbf{p}})}{2}\right) \frac{p_{\max} - p_{\min} + 4\beta + 2\lambda}{p_{\min} + \frac{\lambda}{T}} \right),
\end{equation*}
}

\noindent where $\alpha$ is the competitive ratio of \RORO for \ROCU given in \autoref{thm:roro-comp-rocu}.  

\begin{proof}[Proof of \autoref{thm:uq-uqrobustness}]
To show this result, recall that the point predictions are \textit{imperfect} (i.e., $\hat{\mathbf{p}} \neq \mathbf{p}$), but the uncertainty quantification is \textit{accurate} (i.e., $\mathbf{p} \in \mathcal{U}$).  We start by proving the following lemma that relates the cost of $\ADV(\mathcal{I}) \coloneqq \OPT(\mathbf{\hat{p}})$ to the cost of $\OPT(\mathcal{I})$:
\begin{lemma} \label{lem:adv-opt-gap}
Given the decision uncertainty score $\DUS(\mathcal{U}, \hat{\mathbf{p}})$, we have the following bound on the cost of $\ADV$.
\begin{equation}
    \Cost(\ADV(\mathcal{I})) \leq \Cost(\OPT(\mathcal{I})) + \frac{\DUS(\mathcal{U}, \hat{\mathbf{p}})}{2}(p_{\max} - p_{\min} + 4\beta + 2\lambda).
\end{equation}
\end{lemma}
\begin{proof}
Recall the definition of the decision uncertainty score:
\begin{align*}
    \DUS(\mathcal{U}, \hat{\mathbf{p}}) \coloneqq \max_{\mathbf{z} \in \mathcal{U}} \Vert \OPT(\hat{\mathbf{p}}) - \OPT(\mathbf{z})\Vert_1.
\end{align*}
Since $\mathbf{p} \in \mathcal{U}$ and thus the instance $\mathcal{I}$ lies within the uncertainty set, we know that:
\begin{align*}
    \Vert \OPT(\hat{\mathbf{p}}) - \OPT(\mathcal{I})\Vert_1 \leq \DUS(\mathcal{U}, \hat{\mathbf{p}}).
\end{align*}
Let $\mathbf{\hat{x}} = \OPT(\hat{\mathbf{p}})$ and $\mathbf{x}^* = \OPT(\mathcal{I})$.  Then, we have:
\begin{align*}
    \Vert \mathbf{\hat{x}} - \mathbf{x}^* \Vert_1 = \sum_{t=1}^T |\hat{x}_t - x^*_t| \leq \DUS(\mathcal{U}, \hat{\mathbf{p}}).
\end{align*}
First, we consider the difference in the purchasing cost (i.e., the first term in \eqref{eq:obj}).
There are two cases to consider: (i) $\hat{x}_t \geq x^*_t$ and (ii) $\hat{x}_t < x^*_t$.  
Observe that whenever $\hat{x}_t \geq x^*_t$, \ADV purchases \textit{more} than \OPT, and the worst-case is captured by $p_{\max}(\hat{x}_t - x^*_t)$, since the price $p_t$ can be as high as $p_{\max}$.  
Conversely, whenever $\hat{x}_t < x^*_t$, \ADV purchases \textit{less} than \OPT, and the worst-case is captured by $p_{\min}(x^*_t - \hat{x}_t)$, since the price $p_t$ can be as low as $p_{\min}$.  Thus, we have:
\begin{align*}
    \sum_{t=1}^T p_t \hat{x}_t - \sum_{t=1}^T p_t x^*_t &\leq \sum_{t \in \{\hat{x}_t \geq x^*_t\}} p_{\max} (\hat{x}_t - x^*_t) + \sum_{t \in \{\hat{x}_t < x^*_t\}} p_{\min} (x^*_t - \hat{x}_t) \\
    &\leq p_{\max} \sum_{t \in \{\hat{x}_t \geq x^*_t\}} (\hat{x}_t - x^*_t) + p_{\min} \sum_{t \in \{\hat{x}_t < x^*_t\}} (x^*_t - \hat{x}_t).
\end{align*}
Since we know that $\sum_{t=1}^T \hat{x}_t = 1$ and $\sum_{t=1}^T x^*_t = 1$, the increases in $x_t$'s and the decreases in $x_t$'s must balance out, i.e., we have:
\begin{align*}
    \sum_{t \in \{\hat{x}_t \geq x^*_t\}} (\hat{x}_t - x^*_t) = \sum_{t \in \{\hat{x}_t < x^*_t\}} (x^*_t - \hat{x}_t) = \frac{1}{2} \Vert \mathbf{\hat{x}} - \mathbf{x}^* \Vert_1.
\end{align*}
Thus, we have:
\begin{align*}
    \sum_{t=1}^T p_t \hat{x}_t - \sum_{t=1}^T p_t x^*_t \leq \left( p_{\max} - p_{\min} \right) \frac{\DUS(\mathcal{U}, \hat{\mathbf{p}})}{2}.
\end{align*}
Next, we consider the difference in the switching cost (i.e., the second term in \eqref{eq:obj}).  In the worst-case, the switching cost of \ADV can be higher than that of \OPT by $2\beta$ per unit change in the decisions.  Specifically, in the worst-case, whenever $\hat{x}_t > x^*_t$, \ADV incurs an extra switching cost of $2 \beta(\hat{x}_t - x^*_t)$ compared to \OPT, and whenever $\hat{x}_t < x^*_t$, \ADV incurs an extra switching cost of $2 \beta(x^*_t - \hat{x}_t)$ compared to \OPT.  In aggregate, we have:
\begin{align*}
    \sum_{t=1}^{T+1} \beta |\hat{x}_t - \hat{x}_{t-1}| - \sum_{t=1}^{T+1} \beta |x^*_t - x^*_{t-1}| \leq 2\beta \cdot \DUS(\mathcal{U}, \hat{\mathbf{p}}).
\end{align*}
Finally, we consider the difference in the regularization cost (i.e., the third term in \eqref{eq:obj}).  First, since $\Vert \mathbf{y} \Vert_2$ is a convex and symmetric function on the simplex (recall that the simplex is defined as $\{\mathbf{y} \in \mathbb{R}^T_+ : \sum_{t=1}^T y_t = 1\}$), maximizing the difference $\Vert \mathbf{\hat{x}} \Vert_2 - \Vert \mathbf{x}^* \Vert_2$ is equivalent to ``pushing'' $\mathbf{\hat{x}}$ to an extreme point of the simplex (i.e., closer to a unit vector) and, for that fixed $\mathbf{\hat{x}}$, minimize $\Vert \mathbf{x}^* \Vert_2$ by ``pushing'' $x^*$ towards the center of the simplex (i.e., the vector with all entries equal to $\frac{1}{T}$).  
Concretely, suppose that $\mathbf{\hat{x}}$ is a unit vector (i.e., $\hat{x}_j = 1$ for some $j$ and $\hat{x}_{t'} = 0$ for all $t' \neq j$).  The decision uncertainty score constraint (i.e., $\Vert \mathbf{\hat{x}} - \mathbf{x}^* \Vert_1 \leq \DUS(\mathcal{U}, \hat{\mathbf{p}})$) forces that $x^*_{j}$ must be at least $1 - \frac{\DUS(\mathcal{U}, \hat{\mathbf{p}})}{2}$, and thus the remaining $T-1$ entries of $\mathbf{x}^*$ must sum to at most $\DUS(\mathcal{U}, \hat{\mathbf{p}})$.  To minimize $\Vert \mathbf{x}^* \Vert_2$, we should set these remaining $T-1$ entries to be equal, i.e., $x^*_{t'} = \frac{\DUS(\mathcal{U}, \hat{\mathbf{p}})}{2(T-1)}$ for all $t' \neq j$.  Thus, we have:
\begin{align*}
    \max_{\hat{x}, x^*} \left( \Vert \mathbf{\hat{x}} \Vert_2 - \Vert \mathbf{x}^* \Vert_2 \right) &= \DUS(\mathcal{U}, \hat{\mathbf{p}}) - \frac{T}{4(T-1)}\DUS(\mathcal{U}, \hat{\mathbf{p}})^2,
\end{align*}
and since $\frac{T}{4(T-1)}\DUS(\mathcal{U}, \hat{\mathbf{p}})^2 \geq 0$, we have:
\begin{align*}
    \max_{\hat{x}, x^*} \left( \Vert \mathbf{\hat{x}} \Vert_2 - \Vert \mathbf{x}^* \Vert_2 \right) &\leq \DUS(\mathcal{U}, \hat{\mathbf{p}}),
\end{align*}
Hence, for any $\mathbf{\hat{x}}$ and $\mathbf{x}^*$ that satisfy the decision uncertainty score constraint, we have:
\begin{align*}
    \Vert \mathbf{\hat{x}} \Vert_2 - \Vert \mathbf{x}^* \Vert_2 \leq \DUS(\mathcal{U}, \hat{\mathbf{p}}).
\end{align*}
Multiplying by $\lambda$ yields that the difference in regularization cost is at most $\lambda \cdot \DUS(\mathcal{U}, \hat{\mathbf{p}})$.  

Combining the above three bounds, we have:
\begin{align*}
    \Cost(\ADV(\mathcal{I})) - \Cost(\OPT(\mathcal{I})) &\leq \frac{\DUS(\mathcal{U}, \hat{\mathbf{p}})}{2}(p_{\max} - p_{\min}) + 2\beta \cdot \DUS(\mathcal{U}, \hat{\mathbf{p}}) + \lambda \cdot \DUS(\mathcal{U}, \hat{\mathbf{p}}) \\
    &= \frac{\DUS(\mathcal{U}, \hat{\mathbf{p}})}{2}(p_{\max} - p_{\min} + 4\beta + 2\lambda).
\end{align*}
Rearranging yields the desired result.
\end{proof}

From the proofs of \autoref{thm:uq-consistency} and \autoref{thm:uq-robustness}, we know that the cost of \UQAdv is a convex combination of the cost of \ADV and the cost of \RORO, weighted by $\gamma$ and $1-\gamma$, respectively.  Specifically, we have the following, where we let $r \coloneqq \Vert \mathbf{\hat{x}} \Vert_2^2$ denote the regularization cost paid by \ADV for brevity.
\begin{align*}
    \Cost(\UQAdv(\mathcal{I})) &= \sum_{t=1}^{T} p_t x_t + \sum_{t=1}^{T+1} \beta \vert x_t - x_{t-1} \vert + \lambda \Vert \mathbf{x} \Vert_2^2,\\
    &= \sum_{t=1}^{T} p_t(\gamma \hat x_t + (1-\gamma) \tilde x_t) + \sum_{t=1}^{T+1} \beta|\gamma \hat x_t + (1-\gamma) \tilde x_t - \gamma \hat x_{t-1} - (1-\gamma) \tilde x_{t-1}| + \lambda \Vert \gamma \hat{\mathbf{x}} + (1-\gamma) \tilde{\mathbf{x}} \Vert_2^2, \\
    &\leq \gamma \sum_{t=1}^{T} p_t \hat x_t + (1-\gamma) \sum_{t=1}^{T} p_t \tilde x_t + \sum_{t=1}^{T+1} \beta|\gamma (\hat x_t - \hat x_{t-1}) + (1-\gamma) (\tilde x_t - \tilde x_{t-1})| + \lambda (\gamma^2 r + (1-\gamma)^2 \Vert \tilde{\mathbf{x}} \Vert_2^2 + 2\gamma(1-\gamma) \langle \hat{\mathbf{x}}, \tilde{\mathbf{x}} \rangle) \\
\end{align*}
By triangle inequality, the switching cost term is upper bounded as follows:
\begin{align*}
    \Cost(\UQAdv(\mathcal{I})) &\leq \gamma \sum_{t=1}^{T} p_t \hat x_t + (1-\gamma) \sum_{t=1}^T p_t \tilde x_t + \gamma \sum_{t=1}^{T+1} \beta|\hat x_t + \hat x_{t-1}| + (1-\gamma)\sum \beta | \tilde x_t - \tilde x_{t-1}| + \lambda (\gamma^2 r + (1-\gamma)^2 \Vert \tilde{\mathbf{x}} \Vert_2^2 + 2\gamma(1-\gamma) \langle \hat{\mathbf{x}}, \tilde{\mathbf{x}} \rangle) \\
\end{align*}
Using Cauchy-Schwarz on the final term, the above simplifies to:
\begin{align*}
    \Cost(\UQAdv(\mathcal{I})) &\leq \gamma \sum_{t=1}^{T} p_t \hat x_t + (1-\gamma) \sum_{t=1}^T p_t \tilde x_t + \gamma \sum_{t=1}^{T+1} \beta|\hat x_t + \hat x_{t-1}| + (1-\gamma)\sum \beta | \tilde x_t - \tilde x_{t-1}| + \lambda (\gamma \sqrt{r} + 1 - \gamma)^2 \\
\end{align*}
With a slight abuse of notation, let $\ADV_{\text{bought}}(\mathcal{I})$, $\ADV_{\text{switch}}(\mathcal{I})$, and $\ADV_{\text{reg}}(\mathcal{I}) \coloneqq r \cdot \lambda$ denote the cost of \ADV due to the buying, switching, and regularization costs, respectively.  Similarly, let $\RORO_{\text{bought}}(\mathcal{I})$, $\RORO_{\text{switch}}(\mathcal{I})$, and $\RORO_{\text{reg}}(\mathcal{I}) \coloneqq 1 \cdot \lambda$ denote the (worst-case) cost of \RORO due to the buying, switching, and regularization costs, respectively.  Then, we have:
\begin{align*}
    \Cost(\UQAdv(\mathcal{I})) &\leq \gamma \ADV_{\text{bought}}(\mathcal{I}) + (1-\gamma) \RORO_{\text{bought}}(\mathcal{I}) + \gamma \ADV_{\text{switch}}(\mathcal{I}) + (1-\gamma)\RORO_{\text{switch}}(\mathcal{I}) + \lambda (\gamma \sqrt{r} + (1-\gamma) 1)^2 \\
\end{align*}
By Jensen's inequality, we have:
$$\left( \gamma \sqrt{r} + (1-\gamma) \sqrt{1} \right)^2 \leq \gamma r + (1-\gamma) 1.$$
Combining this with the above gives:
\begin{align*}
    \Cost(\UQAdv(\mathcal{I})) &\leq \gamma (\ADV_{\text{bought}}(\mathcal{I}) + \ADV_{\text{switch}}(\mathcal{I}) + \ADV_{\text{reg}}(\mathcal{I})) + (1-\gamma) (\RORO_{\text{bought}}(\mathcal{I}) + \RORO_{\text{switch}}(\mathcal{I}) + \RORO_{\text{reg}}(\mathcal{I})) \\
    &\leq \gamma \cdot \Cost(\ADV(\mathcal{I})) + (1-\gamma) \cdot \Cost(\RORO(\mathcal{I})).
\end{align*}
By \sref{Lemma}{lem:adv-opt-gap}, we have:
\begin{align*}
    \Cost(\UQAdv(\mathcal{I})) &\leq \gamma \cdot \left( \Cost(\OPT(\mathcal{I})) + \frac{\DUS(\mathcal{U}, \hat{\mathbf{p}})}{2}(p_{\max} - p_{\min} + 4\beta + 2\lambda) \right) + (1-\gamma) \cdot \alpha \cdot \Cost(\OPT(\mathcal{I})), \\
    &= \left( \gamma + (1-\gamma) \alpha \right) \cdot \Cost(\OPT(\mathcal{I})) + \gamma \cdot \frac{\DUS(\mathcal{U}, \hat{\mathbf{p}})}{2}(p_{\max} - p_{\min} + 4\beta + 2\lambda). 
\end{align*}
Note that in the worst case, the second term contributes maximally to the UQ robustness ratio when $\Cost(\OPT(\mathcal{I}))$ is minimized, i.e., when $\Cost(\OPT(\mathcal{I})) = p_{\min} + \frac{\lambda}{T}$.  Thus, we have:
\begin{align*}
    \Cost(\UQAdv(\mathcal{I})) &\leq \left( \gamma + (1-\gamma) \alpha \right) \cdot \Cost(\OPT(\mathcal{I})) + \gamma \cdot \frac{\DUS(\mathcal{U}, \hat{\mathbf{p}})}{2} \left( \frac{p_{\max} - p_{\min} + 4\beta + 2\lambda}{p_{\min} + \frac{\lambda}{T}} \right) \cdot \Cost(\OPT(\mathcal{I})). 
\end{align*}
In terms of \DUS, this is equivalently stated as:
\begin{align}
    \Cost(\UQAdv(\mathcal{I})) &\leq \left( 1 + \left( \frac{\DUS(\mathcal{U}, \hat{\mathbf{p}})}{2} \right) ( \alpha -1 ) \right)  \Cost(\OPT(\mathcal{I})) + \left[ \frac{\DUS(\mathcal{U}, \hat{\mathbf{p}})}{2} - \frac{\DUS(\mathcal{U}, \hat{\mathbf{p}})^2}{4} \right]  \left( \frac{p_{\max} - p_{\min} + 4\beta + 2\lambda}{p_{\min} + \frac{\lambda}{T}} \right)  \Cost(\OPT(\mathcal{I})). 
\end{align}

Dividing both sides by $\Cost(\OPT(\mathcal{I}))$ yields:
\begin{align*}
    \frac{\Cost(\UQAdv(\mathcal{I}))}{\Cost(\OPT(\mathcal{I}))} &\leq  1 + \left( \frac{\DUS(\mathcal{U}, \hat{\mathbf{p}})}{2} \right) \left( \alpha - 1 + \left( 1 - \frac{\DUS(\mathcal{U}, \hat{\mathbf{p}})}{2}\right) \frac{p_{\max} - p_{\min} + 4\beta + 2\lambda}{p_{\min} + \frac{\lambda}{T}} \right). \\
\end{align*}
This completes the proof.

\end{proof}

\end{document}